\documentclass[journal]{IEEEtran}

\usepackage[titles]{tocloft}

\usepackage{comment}
\usepackage{braket}
\usepackage{amsmath,amssymb,amsfonts,dsfont,amsthm}
\usepackage{revsymb}
\usepackage{algorithmic}
\usepackage{graphicx}
\usepackage{textcomp}
\usepackage{xcolor}
\usepackage[colorlinks=true,linkcolor=black,anchorcolor=black,citecolor=black,filecolor=black,menucolor=black,runcolor=black,urlcolor=black]{hyperref}
\usepackage{orcidlink}

\usepackage[style=ieee,giveninits=false,sorting=none,minnames=3,maxnames=10,dashed=false]{biblatex}

\newtheorem{theorem}{Theorem}
\newtheorem{corollary}[theorem]{Corollary}
\newtheorem{lemma}[theorem]{Lemma}
\newtheorem{definition}[theorem]{Definition}

\newtheorem{example}[theorem]{Example}
\newtheorem{remark}[theorem]{Remark}

\newcommand\myfrac[2]{\genfrac{}{}{0pt}{}{#1}{#2}}

\newlength{\blank}
\newcommand{\EE}{\mathbb{E}}
\newcommand{\RR}{\mathbb{R}}

\newcommand{\cB}{\mathcal{B}}
\newcommand{\cC}{\mathcal{C}}

\newcommand{\cE}{\mathcal{E}}
\newcommand{\cL}{\mathcal{L}}

\newcommand{\cN}{\mathcal{N}}
\newcommand{\cP}{\mathcal{P}}
\newcommand{\cR}{\mathcal{R}}
\newcommand{\cS}{\mathcal{S}}
\newcommand{\cT}{\mathcal{T}}

\newcommand{\cX}{\mathcal{X}}
\newcommand{\cY}{\mathcal{Y}}

\newcommand{\1}{\openone} 

\newcommand{\ox}{\otimes}

\DeclareMathOperator{\Tr}{Tr}

\DeclareFontFamily{OMX}{yhex}{}
\DeclareFontShape{OMX}{yhex}{m}{n}{<->yhcmex10}{}
\DeclareSymbolFont{yhlargesymbols}{OMX}{yhex}{m}{n}
\DeclareMathAccent{\wideparen}{\mathord}{yhlargesymbols}{"F3}

\begin{document}

\title{Rate-reliability tradeoff\protect\\ for deterministic identification\thanks{This paper has been accepted for publication at IEEE Transactions on Communications. A preliminary version of the present work was presented at the 2025 IEEE International Conference on Communications, Montreal QB, 8-12 June 2025 \cite{CDBW:Reliability_ICC}; and the recent results session of the 2025 IEEE International Symposium on Information Theory, Ann Harbor MI (USA), 22-27 June 2025.}\thanks{HB and CD are supported by the German Federal Ministry of Research, Technology and Space (BMFTR) within the national initiative on 6G Communication Systems through the research hub 6G-life, grants 16KISK002 and 16KISK263, within the national initiative on Post Shannon Communication (NewCom), grants 16KIS1003K and 16KIS1005, within the national initiative ``QuaPhySI'', 
grants 16KISQ1598K and 16KIS2234, and within the national initiative ``QTOK'', 
grant 16KISQ038. 
HB has further received funding from the German Research Foundation (DFG) within Germany’s Excellence Strategy EXC-2092-390781972. 
AW is supported by the European Commission QuantERA grant ExTRaQT (Spanish MICIN project PCI2022-132965), by the Spanish MICIN (projects PID2019-107609GB-I00 and PID2022-141283NB-I00) with the support of FEDER funds, by the Spanish MICIN with funding from European Union NextGenerationEU (PRTR-C17.I1) and the Generalitat de Catalunya, and by the Alexander von Humboldt Foundation.
PC and AW are furthermore supported by the Institute for Advanced Study of the Technical University Munich (IAS-TUM).}}
\author{Pau~Colomer\textsuperscript{\orcidlink{0000-0002-0126-4521}}\thanks{P. Colomer (pau.colomer@tum.de) is with the Chair of Theoretical Information Technology, Technische Universit\"at M\"unchen (LTI-TUM), Theresienstra{\ss}e 90, 80333 M\"unchen, Germany; and IAS-TUM, Lichtenbergstra{\ss}e 2a, 85748 Garching, Germany.},~\IEEEmembership{Student Member,~IEEE}, 
Christian~Deppe\textsuperscript{\orcidlink{0000-0003-3047-3549}}\thanks{C. Deppe (christian.deppe@tu-bs.de) is with the Institute for Communications Technology and the 6G-life research hub in Technische Universit\"at Braunschweig, Schleinitzstra{\ss}e 22, 38106 Braunschweig, Germany 
.},~\IEEEmembership{Senior Member,~IEEE},\protect\\
Holger~Boche\textsuperscript{\orcidlink{0000-0002-8375-8946}}\thanks{H. Boche (boche@tum.de) is with LTI-TUM and the 6G-life research hub in Theresienstra{\ss}e 90, 80333 M\"unchen, Germany; Munich Center for Quantum Science and Technology, Schellingstra{\ss}e 4, 80799 München, Germany; and Munich Quantum Valley, Leopoldstra{\ss}e 244, 80807 München, Germany.},~\IEEEmembership{Fellow,~IEEE},
and~Andreas~Winter\textsuperscript{\orcidlink{0000-0001-6344-4870}}\thanks{A. Winter (andreas.winter@uni-koeln.de) is with Department Mathematik/ Informatik--Abteilung Informatik, Universit\"at zu K\"oln, Albertus-Magnus-Platz, 50923 K\"oln, Germany, as well as with IAS-TUM. He was previously with ICREA and Grup d'Informaci\'o Qu\`antica, Departament de F\'isica, Universitat Aut\`onoma de Barcelona, 08193 Bellaterra (Barcelona), Spain.}}
\maketitle

\thispagestyle{empty}

\begin{abstract}
We investigate deterministic identification over arbitrary memoryless channels under the constraint that the error probabilities of first and second kind are exponentially small in the block length $\mathbf{n}$, controlled by reliability exponents $\mathbf{E_1,E_2 \geq 0}$. In contrast to the regime of slowly vanishing errors, where the identifiable message length scales linearithmically as $\mathbf{\Theta(n\log n)}$, here we find that for positive exponents linear scaling is restored, now with a rate that is a function of the reliability exponents.
We give upper and lower bounds on the ensuing rate-reliability function in terms of (the logarithm of) the packing and covering numbers of the channel output set, which for small error exponents $\mathbf{E_1,E_2>0}$ can be expanded in leading order as the product of the Minkowski dimension of a certain parametrisation the channel output set and $\mathbf{\log\min\{E_1,E_2\}}$. These allow us to recover the previously observed slightly superlinear identification rates, and offer a different perspective for understanding them in more traditional information theory terms.
We also show that even if only one of the two errors is required to be exponentially small, the linearithmic scaling is lost. 
We further illustrate our results with a discussion of the case of dimension zero, and extend them to classical-quantum channels and quantum channels with tensor product input restriction.
\end{abstract}



\section{Introduction}
\label{sec:intro}
\subsection{Background and previous work}\label{ssec:background}


\IEEEPARstart{I}{n} Shannon's fundamental model of communication \cite{Shannon:TheoryCommunication}, the receiver aims to faithfully recover an original message sent through $n$ uses of a memoryless noisy channel $W:\cX\rightarrow\cY$, given by a probability transition kernel between measurable spaces $\cX$ and $\cY$. The maximum number $M$ of transmittable messages scales exponentially with $n$, and thus we define the \emph{(linear-scale) rate} of a transmission code as $R=\frac1n\log M$, with $\log$ the binary logarithm. The maximum rate for asymptotically faithful transmission -- that is, with arbitrarily small probability of error for sufficiently long code words -- is known as the \emph{capacity} $C$ of the channel \cite{Shannon:TheoryCommunication,Wolfowitz:converse}. 

This model contrasts with the \emph{identification} task in which instead of recovering the initial message, the receiver is only interested in knowing whether it is equal to one particular message of their choice or not. 
Jaja demonstrated in \cite{Ja:ID_easier} that this task involves lower communication complexity than Shannon's original transmission scenario. Building on this result, Ahlswede and Dueck \cite{AD:ID_ViaChannels} fully characterized the identification problem and analyzed its performance over general discrete memoryless channels (DMCs) using a randomized encoder.
They showed that identification codes can achieve doubly exponential growth of the number $N$ of messages as a function of the block length $n$, i.e.~$\log N \sim 2^{nR}$. In other words, we can identify exponentially more messages than we can transmit. 
In general, randomized identification codes 
can achieve this exponential scaling \cite{Han:book,LDB:mc}.
As these results rely crucially on randomized encodings, it is natural to ask what happens if we impose a deterministic approach.

\begin{definition}
\label{def:DI code}
An $(n,N,\lambda_1,\lambda_2)$ \emph{deterministic identification (DI) code} is a family $\{(u_j,\cE_j) : j\in[N]\}$ consisting of code words $u_j\in\cX^n$ and subsets $\cE_j\subset \cY^n$, such that 
\begin{equation}
    \forall j\neq k\in[N], \quad
    W_{u_j}(\cE_j) \geq 1-\lambda_1,\quad 
    W_{u_j}(\cE_k) \leq \lambda_2,
\end{equation}
where we use the notation $W_{u_j}:=W^n(\cdot|u_j)$ for the conditional probability distributions at the output. 
\end{definition}

Notice the structural differences between transmission and identification codes: in transmission, there is a single type of error (incorrect message decoding) while in identification, two distinct errors arise. A missed identification $\lambda_1$ occurs when the correct message is not identified, and a wrong identification $\lambda_2$ when an incorrect message is mistakenly identified. The latter is directly tied to the fact that identification decoding sets $\cE_j$ may overlap, unlike the decoding regions required in transmission which are disjoint.

It was initially observed that DI leads to much poorer results in terms of code size scaling in the block length \cite{AD:ID_ViaChannels, AC:DI} than the unrestricted randomized identification. Indeed, DI over discrete memoryless channels can only lead to linear scaling of the message length, $\log N \sim Rn$ \cite{SPBD:DI_power} as in Shannon's communication paradigm (albeit with a higher rate $R$). 
Despite this poorer performance, interest in deterministic codes has recently been renewed, as they have proven to be easier to implement and simulate \cite{DI_simpler_impl}, to explicitly construct \cite{DI_explicit_construction}, and offer reliable single-block performance \cite{AD:ID_ViaChannels}; see in particular \cite{VDTB:practical-DI} for identification codes.
Also, surprisingly, certain channels with continuous input alphabets have DI codes governed by a slightly superlinear scaling in the block length: $\log N\sim Rn\log n$ (commonly referred to as \emph{linearithmic}). 
This was first observed for Gaussian channels \cite{SPBD:DI_power}, Gaussian channels with both fast and slow fading \cite{DI-fading,VDB:DI-fading-new}, and Poisson channels \cite{DI-poisson,DI-poisson_mc}. In our recent work \cite{CDBW:DI_classical,CDBW:DI_proceedings} we have generalized all these particular DI results, proving that the linearithmic scaling observed in these special cases is actually a feature exhibited by general channels with discrete output. For the rest of the paper we will in fact assume that $\cY$ is finite.

It is important to emphasize the fundamental implications of different scalings beyond the value of the rate expressions. When a code achieves superlinear scaling -- for example linearithmic $\log N(n) \sim R n \log n$ --, the number of identifiable messages grows \emph{asymptotically faster} than in the linear regime $\log N(n) \sim R'n$, regardless of how large the linear rate constant might be. That is, even if a linear code achieves a higher rate $R'>R$, the superlinear code will eventually yield \emph{infinitely more} identifiable messages as $n \to \infty$. 
This distinction justifies a careful attention to the growth regime. Specifically when calculating the capacity of a channel, it is crucial to define it in the correct scaling. Indeed, if we have a linearithmic growth of code length, the linear (exponential) scale rate will trivially be $\infty$ ($0$). To define a meaningful capacity value, it is essential to choose the suitable scale -- in the present case linearithmic.

In \cite{CDBW:DI_classical}, linearithmic DI capacities $\dot{C}_\text{DI}$ of general channels were analysed in the pessimistic and optimistic settings; the formal definition of the capacities can be found in Subsection \ref{ssec:linear_to_super} (for further insight on these settings see \cite{Ahlswede2006} or the discussion in \cite[Sect.~7.1]{CDBW:DI_classical}).
It was shown that they can be respectively bounded in terms of the lower and upper covering dimension $\underline{d}_M$ and $\overline{d}_M$ [an introduction and the definitions of these concepts is found in Section \ref{sec:preliminaries} below, see Equations \eqref{eq:lower_d_def} and \eqref{eq:upper_d_def}] of a certain algebraic transformation of the set of output probability distributions of the channel: 
\begin{align}
  \frac14\underline{d}_M
    &\leq \dot{C}_{\text{DI}}(W)\,\,
    \leq \liminf_{n\rightarrow\infty} \frac{1}{n\log n}\log N\,
    \leq \frac12\underline{d}_M,\label{eq:lower_d_def}\\
    \frac14\overline{d}_M
    &\leq \dot{C}_{\text{DI}}^{\text{opt}}(W) 
    \leq \limsup_{n\rightarrow\infty}\frac{1}{n\log n}\log N
    \leq \frac12\overline{d}_M.\label{eq:upper_d_def}
\end{align}
When $d=0$, such as in the case of the discrete memoryless channel, a more general and abstract theorem is provided to meaningfully bound the size of the code at the correct (in general slower, for DMCs in fact linear) scaling \cite[Thm.~5.11]{CDBW:DI_classical}.

\subsection{Motivation and preliminary problem setting}
\label{ssec:motivation}
While the DI capacity is defined as the maximum achievable (linearithmic) rate for reliable identification -- where the error probability can be made arbitrarily small with sufficiently large $n$ --, it does a priori not provide information on how quickly the errors vanish.
In contrast, it is well-known in the communication setting \cite{Shannon:TheoryCommunication,Gallager:ReliabilityBook} that at (linear-scale) rates below capacity, $0\leq R < C$, the error probability decreases exponentially with the block length: $\lambda \sim 2^{-En}$. 
In the present paper we initiate the study of the performance of DI codes when $\lambda_1$ and/or $\lambda_2$ are exponentially small, in an attempt to better understand the nature of the superlinear rates.

Anticipating our subsequent results, we associate to every $(n,N,\lambda_1,\lambda_2)$-DI code (recall Definition \ref{def:DI code}) the triple of linear-scale message and error rates $(R(n),E_1(n),E_2(n))$, where:
\begin{equation}
\!\!\!\!R(n)\! :=\! \frac1n\log N,\,\,
  E_1(n)\! :=\! \frac1n\log\frac{1}{\lambda_1},\,
  E_2(n)\! :=\! \frac1n\log\frac{1}{\lambda_2}.
\end{equation}
The main object of study in this work is the tradeoff between the linear-scale rate functions $R(n)$ and the error exponents $E_{1,2}(n)$ (with the notation $E_{1,2}$ meaning that the same equation holds for both $E_1$ and $E_2$), first in the finite block length regime, and then on the asymptotic limit $n\rightarrow\infty$. The error exponents can of course be functions of $n$ as long as $\lim_{n\rightarrow\infty}nE_{1,2}(n)=\infty$, since we need the errors to vanish in the asymptotic limit; in other words, $E_{1,2}(n)\geq\omega(1/n)$.


We call $(R,E_1,E_2)$ an \emph{(asymptotically) achievable triple} if there exist DI codes with $(R(n),E_1(n),E_2(n))$ converging to $(R,E_1,E_2)$ for $n\rightarrow\infty$. The set of all achievable triples will be referred to as the \emph{achievable region} $\cR(W):=\{(R,E_1,E_2) \text{ achievable} \}$. The rate-reliability function is then
\begin{equation}
R(E_1,E_2) := \sup_{E_1,E_2}R \quad \text{s.t.} \quad (R,E_1,E_2)\in\cR(W).
\end{equation}
Similarly, we can define an \emph{optimistic achievable region} $\overline{\cR}(W)\supseteq\cR(W)$ by allowing the convergence to happen only for a subsequence $n_k\rightarrow\infty$ of block lengths. That is to say, $\overline{\cR}(W)$ is the set of all triples $(R,E_1,E_2)$ such that there exist deterministic identification codes of block length $n_k\rightarrow\infty$ with $(R{(n_k)},E_1{(n_k)},E_2{(n_k)}) \rightarrow (R,E_1,E_2)$. The optimistic rate-reliability function is thus
\begin{equation}
\overline{R}(E_1,E_2) := \sup_{E_1,E_2} R \quad \text{s.t.} \quad (R,E_1,E_2)\in\overline{\cR}(W).
\end{equation}
Note that both $\cR(W)$ and $\overline{\cR}(W)$ are so-called \emph{corners} in $\RR^3$: they are both actually contained in the positive orthant $\RR_{\geq 0}^3$, contain the origin, and with every point $(R,E_1,E_2)$ contain the entire box $[0;R]\times[0;E_1]\times[0;E_2]$.

While in our previous work \cite{CDBW:DI_classical} we had no restrictions on the identification errors other than them being sufficiently small as $n\rightarrow\infty$, here we study in depth the performance of codes with different error settings. We start with the case where both errors vanish exponentially, for which we provide lower and upper bounds on the trade-off between rate and error exponents in the finite block length regime. The study of these bounds in the asymptotic limit allow us to show not only that codes with exponentially vanishing errors can only have linear scaling, but also to find the error conditions under which the main superlinear results in \cite{CDBW:DI_classical} can be recovered. We furthermore study the cases where only one of the two errors is constrained to vanish exponentially, showing that linearithmic scaling is also lost. Despite the current analysis being much more complete both in terms of error settings and non-asymptotic regimes, we can rely on many ideas and tools from our previous work \cite{CDBW:DI_classical}, as it was written with an eye to finite block length. 

\begin{figure}[ht]
    \centering
    \includegraphics[width=0.9\linewidth]{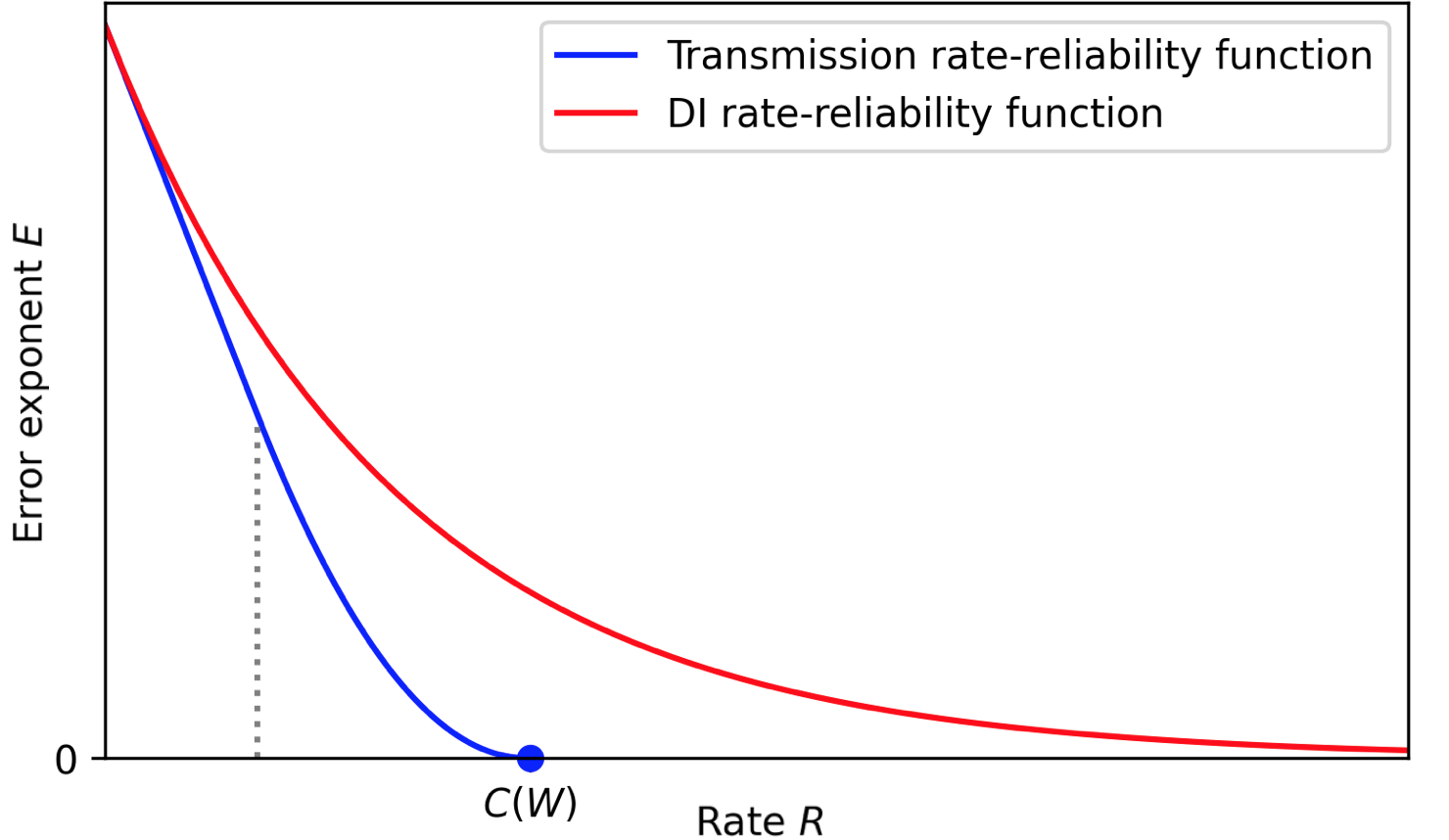}
    \caption{\small Schematic of the rate-reliability function for transmission and 
    DI. The rate-reliability function for transmission (blue line) is monotonically non-increasing with different regions depending on the optimum coding strategies (see the tendency change from linear to quadratic behaviour marked by the grey dotted line). For error exponents going to zero (very slow error decrease with the block length $n$), the linear rate converges to the transmission capacity $C(W)$; no higher rates are possible. Here, in Sections \ref{sec:bounding} and \ref{sec:discussion}, we provide bounds on the rate-reliability function for DI (red line) that show how the rate (defined in the linear scale) goes to infinity logarithmically in the error exponents $E_1 = E_2 = E \approx 0$, corresponding to the linearithmic code length in that regime. Note that the DI curve must dominate the transmission curve everywhere, since an $(n,N,\lambda)$-transmission code is automatically an $(n,N,\lambda,\lambda)$-DI code -- indeed, if we can decode a message, we can identify it.}
    \label{fig:T_vs_DI}
\end{figure}

\subsection{Outline}\label{ssec:outline}
After giving the necessary technical tools in Section \ref{sec:preliminaries}, we present upper and lower bounds on the non-asymptotic rates as functions of the error exponents in Section \ref{sec:bounding}.
We remark already here that while these bounds are stated for general exponents, they are most suitable for exponents below a certain threshold. A full account of this fact and alternative analyses for the regime of large error exponents are given in the subsequent sections.

In Section \ref{sec:discussion}, we discuss the bounds from three different angles: in \ref{ssec:linear_to_super} we show that for all (sufficiently small) $E_1,E_2>0$, the rate-reliability functions $R(E_1,E_2)$ and $\overline{R}(E_1,E_2)$ are bounded away from $0$ and $\infty$; and by letting $E_1(n)$ and $ E_2(n)$ vanish slowly, we recover the main linearithmic capacity results from \cite{CDBW:DI_classical} (see the schematic Figure \ref{fig:T_vs_DI}). In Subsection \ref{ssec:d=0} we show how to deal with channels characterised by dimension zero, where the previous bounds seem not to give the full picture of what is happening. 
Finally, in \ref{ssec:large_exponents} we discuss the regime of large error exponents (quickly vanishing errors) and why the bounds in \ref{sec:bounding} are not suitable there.

Section \ref{sec:steins} studies the case of asymmetric error exponents, that is, when we constrain one of the errors to vanish exponentially while the other is just bounded. While the results are used to broaden the previous discussion, we need a completely different approach, hence the separate section.

Finally, in Section \ref{sec:quantum}, we generalize the main bounds in Section \ref{sec:bounding} to the quantum setting for classical-quantum channels, and general quantum channels under the restriction that only product states are used on the input. We close with a brief discussion and open problems in Section \ref{sec:conclusions}.

\section{Preliminaries}
\label{sec:preliminaries}
As we have already mentioned, a number of technical ideas in this work rely on the methods developed in \cite{CDBW:DI_classical}. These are mainly based on metric studies of the output probability sets. The essential distance measures needed are defined next: the \emph{total variation distance} is a statistical distance measure which coincides with half the $L^1$ distance between the probability mass functions. Let $P$ and $Q$ be two probability distributions defined on a finite or countably infinite measurable space $\cL$, then the total variation distance is defined as
\begin{equation}\label{eq:TVD}
\!\!\!\frac{1}{2}\|P-Q\|_1 
    \!:=\sup_{L\in\cL}|P(L)-Q(L)|\!=\!\!\sum_{\ell\in\cL}\frac{1}{2} |P(\ell)-Q(\ell)|.
\end{equation}
The \emph{Bhattacharyya coefficient} (in quantum information called \emph{fidelity}) is given by $F(P,Q)=\sum_{\ell\in\cL}\sqrt{P(\ell)Q(\ell)}$, and it is related to the total variation distance by the following bounds:
\begin{equation}
  \label{eq:Classical_FvdG}
  1-F(P,Q) \leq \frac12 \|P-Q\|_1 \leq \sqrt{1-F(P,Q)^2}.
\end{equation}
Finally, following \cite[Sect.~3]{CDBW:DI_classical} we use a metric on a modified output set that comes out naturally for product distributions, and will allow us to work in a Euclidean space (enabling the use of the Minkowski dimensions which already appeared in Subsection \ref{ssec:background} above, and are defined below). For a probability distribution $P\in\cP(\cL)$, define the unit vector $\sqrt{P}:=(\sqrt{P(\ell)}:\ell\in\cL) \in \RR^{\cL}$. The image of $\cP(\cL)$ under the square root map is the non-negative orthant sector of the unit hypersphere in $\RR^{|\cL|}$, which we shall denote $S_+(|\cL|,1)$. 
As proven in \cite{CDBW:DI_classical}, these objects obey the relation
\begin{equation}
  \label{eq:fidelity->euclidean}
  1-F(P,Q)^2 
    \leq \left| \sqrt{P}-\sqrt{Q} \right|_2^2 
    \leq 2\left(1-F(P,Q)^2\right).
\end{equation}
For a channel $W:\cX \rightarrow \cY$, we denote the output probability set $\widetilde{\cX} := W(\cX) \subset \cP(\cY)$, with $W_x$ the output distribution conditioned on an input $x$ as described in Definition~\ref{def:DI code}. Hence, we have
\begin{equation}
  \sqrt{\!\widetilde{\cX}} 
    = \left\{ \sqrt{W_x} : x\in\cX \right\} 
  \subset S_+(|\cY|,1).
\end{equation}

For the code construction, the \emph{(entropy) conditional typical set} in the output is used as the identification test $\cE_j:=\mathcal{T}_{u_j}^\delta\subset\cY^n$ of the code word $u_j$. For each $x^n\in\cX^n$ it is defined as
\begin{equation}
  \label{eq:entropy-typical-set}
  \!\!\mathcal{T}_{x^n}^\delta \! 
    := \left\{ y^n\in\mathcal{Y}^n \!: 
           \left| \log W_{x^n}(y^n) + H(W_{x^n}) \right| 
           \leq \delta\sqrt{n} \right\}.
\end{equation}
The nice properties exhibited by typical sets (they collect almost all the probability, which is almost evenly distributed among their elements, cf.~\cite[Lemmas~I.11~{\&}~I.12]{winter:PhDThesis}) allow us to bound the errors of first and second kinds (when choosing $\mathcal{T}_{u_j}^\delta$ as decision rule) through the following lemmas, which will be instrumental in the next section.

\begin{lemma}[{\cite[Lemma~2.1]{CDBW:DI_classical}}]
\label{lemma:error1}
Let $K(d) := (\log\max\{d,3\})^2$. Then, for an arbitrary channel $W:\mathcal{X}\rightarrow\mathcal{Y}$, arbitrary block length $n$, $0<\delta\leq \sqrt{n}\log|\cY|$, and for any $x^n\in\mathcal{X}^n$, we have
\begin{equation}
  W_{x^n}(\mathcal{T}_{x^n}^\delta) 
   \geq 1-2\exp\left(-\delta^2/36 K(|\mathcal{Y}|)\right).
\end{equation}
\end{lemma}

\begin{lemma}[{\cite[Lemma~3.1]{CDBW:DI_classical}}]\label{lemma:hyp_test_error2}
Let $x^n,{x'}^n \in \mathcal{X}^n$ be two input sequences such that the corresponding output probability distributions satisfy
$1-\frac12\left\|W_{x^n}-W_{{x'}^n}\right\|_1 \leq \epsilon$. Then,
\begin{equation}\begin{split}
    W_{{x'}^n}(\mathcal{T}_{x^n}^\delta) 
    &\leq 2\exp\left(-\delta^2/36K(|\mathcal{Y}|)\right)\\
    &\phantom{\leq.} + \epsilon\left(1 + 2^{2\delta\sqrt{n}}2^{H(W_{x^n})-H(W_{{x'}^n})}\right).
\end{split}
\end{equation}
\end{lemma}

For both the coding and converse part, we will need the packing and covering of general sets in arbitrary dimensions. For a non-empty bounded subset $F$, the \emph{covering} problem consists in finding the minimum number $\Gamma_\delta(F)$ of closed balls of radius $\delta$ centered at points in $F$ such that their union contains $F$; in contrast, the \emph{packing} problem consists in finding the maximum number $\Pi_\delta(F)$ of pairwise disjoint open balls of radius $\delta$ centered at points in $F$. 
The packing and covering numbers are not equal in general, but they exhibit the same asymptotic behaviour in the exponent as $\delta \to 0$ \cite{Falconer:fractal}. These are fundamental objects in geometry as they can be used to define the \emph{Minkowski dimension} (a.k.a covering, Kolmogorov, or entropy dimension) of a subset in Euclidean space as 
\begin{equation}\label{eq:Minkowski_Dimension}
d_M(F) = \lim_{\delta\rightarrow0} \frac{\log \Gamma_\delta(F)}{-\log\delta} 
  = \lim_{\delta\rightarrow0} \frac{\log \Pi_\delta(F)}{-\log\delta}.
\end{equation}
While for smooth manifolds the topological and Minkowski dimensions coincide, the latter is especially relevant in theory of fractals, as it captures the geometric complexity of the set beyond the scope of the classical topological dimension, see the excellent textbook \cite{Falconer:fractal}, as well as \cite{Robinson:dimensions,Fraser:dimensions} for further insight. For even more general and irregular sets, it is quite common that the above limit [Eq.~\eqref{eq:Minkowski_Dimension}] does not exist. In that case, we define the \emph{upper} and \emph{lower Minkowski dimensions} through the limit superior and limit inferior, respectively: 
\begin{align}
 \overline{d}_M(F) &:= \limsup_{\delta\rightarrow0}\frac{\log \Gamma_\delta(F)}{-\log\delta} 
  = \limsup_{\delta\rightarrow0} \frac{\log \Pi_\delta(F)}{-\log\delta}, \\
 \underline{d}_M(F) &:= \liminf_{\delta\rightarrow0}\frac{\log \Gamma_\delta(F)}{-\log\delta}
  = \liminf_{\delta\rightarrow0} \frac{\log \Pi_\delta(F)}{-\log\delta}.\label{eq:lowerMinkowski}
\end{align}
The dimensions of the image of output probability set $\widetilde{\cX}$ and its square root $\sqrt{\!\widetilde{\cX}}$ are subject to the inequalities
$d_M(\widetilde{\cX}) 
 \leq d_M\!\left(\!\sqrt{\!\widetilde{\cX}}\right) 
 \leq 2d_M(\widetilde{\cX})$.
This is because $\widetilde{\cX}$ is a Lipschitz-continuous image of $\sqrt{\!\widetilde{\cX}}$; and conversely, $\sqrt{\!\widetilde{\cX}}$ is a $\frac12$-H\"older-continuous image of $\widetilde{\cX}$, cf.~\cite{Robinson:dimensions,Fraser:dimensions}.

\section{Bounding the rate-reliability function}
\label{sec:bounding}
We divide this section into three parts. In \ref{ssec:coding}, we present a lower bound on the rate-reliability function for DI, expressed in terms of the packing number of a transformed output probability set, where the radius is determined by the error exponents. This is followed by an analysis of the bound in the regime of small error exponents. Similarly, Subsection \ref{ssec:converse} provides the upper bound using a covering number, along with its analysis for small error exponents. Finally, in \ref{ssec:improved} we analyse the rate bounds when the error exponents are restricted to certain suitable subsets in the vicinity of $0$, raising our lower bound and decreasing our upper bound in the regime of small errors, and illuminating the separate roles of the lower and upper Minkowski dimensions as the objects of merit in pessimistic and optimistic approaches respectively, beyond the capacity value (in the asymptotic limit), for which this separation was first observed in \cite{CDBW:DI_classical}.

\subsection{Lower bound on the rate-reliability function (coding)}
\label{ssec:coding}
We present our lower bound in Theorem \ref{thm:coding}, and identify the regime where it is more tight in Remark~\ref{remark:coding}. In Corollary~\ref{cor:coding}, we analyse further the regime of small error exponents. Finally, 
we include an intuitive discussion on the meaning of the bounds.
From now on we will use the shorthand $c=1/36K(|\cY|)$. 

\begin{theorem}
\label{thm:coding}
For any $0<t<1$, and error exponent function $E(n)>0$, there exists a DI code on block length $n$ with error exponents $E_1(n)\geq E(n)-\frac1n$ and $E_2(n)\geq E(n)-\frac3n$, and its rate lower-bounded 
\begin{equation}
\label{eq:coding_rate_final}
\begin{split}
    R(n) &\geq (1-t)\log\left[\Pi_{\sqrt[4]{\frac{6E(n)}{ct^2}}} \left(\!\sqrt{\!\widetilde{\cX}}\right)\right]\\
    &\phantom{\geq.}-H(t,1-t)-O\left(\frac{\log n}{n}\right),
\end{split}
\end{equation}
with $H(t,1-t)$ the binary entropy of the distribution $\{t,1-t\}$.
\end{theorem}
\begin{proof}
Inspired by the code construction in \cite{CDBW:DI_classical}, we use the typical sets of input strings as decoding elements, for which we know that Lemmas \ref{lemma:error1} and \ref{lemma:hyp_test_error2} give us upper bounds on the errors of first and second kind. We want to impose exponentially decreasing errors, construct a code, and then analyze the resulting rate. Looking at Lemma \ref{lemma:error1} it is immediate to see that for the error of first kind, the exponential decrease follows by letting $\delta:=\tau\sqrt{n}$, so that
\(
  \lambda_1 \leq 2\exp(-c\tau^2 n).
\)

To bound the error of second kind, additional effort is required. We want to use Lemma \ref{lemma:hyp_test_error2}, which requires a minimum separation $1-\frac12\|W_{u_j}-W_{u_k}\|_1 \leq \epsilon$ between any two different code words $u_j\neq u_k\in\cX^n$. That is, the output probability distributions generated by code words need to form a packing with pairwise total variation distance being almost exponentially close to $1$. 
Using the left inequality in Equation~\eqref{eq:Classical_FvdG} first, and the multiplicativity of the fidelity under tensor products after, we have
\begin{equation}
  1-\frac12\|W_{x^n}-W_{{x'}^n}\|_1
   \leq F\left(W_{x^n},W_{{x'}^n}\right)
   =    \prod_{i=1}^n F\left(W_{x_i},W_{x_i'}\right).
\end{equation}
Following \cite{CDBW:DI_classical}, we study the negative logarithm of the expression above. Letting $u_j=x^n=x_1\dots x_n$ and $u_k=x'^n=x'_1\dots x'_n$, we have
\begin{equation}\begin{split}
  \label{eq:coding_TVD_to_Euc}
  -\ln\left(\! 1-\frac12\|W_{x^n}-W_{{x'}^n}\|_1\! \right)\! 
   &\geq \sum_{i=1}^n -\ln F\left(W_{x_i},W_{x_i'}\right) \\
   &
    =    \frac12 \sum_{i=1}^n -\ln F\left(W_{x_i},W_{x_i'}\right)^2 \\
   &
    \geq \frac12 \sum_{i=1}^n 1-F\left(W_{x_i},W_{x_i'}\right)^2 \\
   &
    \geq \frac14 \sum_{i=1}^n \left| \sqrt{W_{x_i}}-\sqrt{W_{x_i'}}\right|_2^2. \\
\end{split}\end{equation}
Where the second line follows from basic logarithm properties, the third from the leading order of the logarithm when the fidelity is close to 1, and the last line from Equation~\eqref{eq:fidelity->euclidean}.
Notice the important simplification that this achieves: before, we had to construct and analyze a packing of full sequences of length $n$ in total variation distance, and now we just need a letter-wise Euclidean packing to ensure a distance between output probability distributions exponentially close to 1. Let $\cX_0$ be a maximum-size Euclidean packing in $\sqrt{\!\widetilde{\cX}}$ of radius $\beta$ and $\Pi_{\beta} := \Pi_{\beta}\left(\!\sqrt{\!\widetilde{\cX}}\right) = |\cX_0|$ the corresponding packing number.

Next we choose a maximum-size code $\cC_t\subset\cX_0^n$ such that the letters are elements of the packing and there is a minimum Hamming distance $d_H(x^n,x'^n)>tn$ between code words $x^n\neq x'^n\in\cC_t$. Then, in the last sum in Equation~\eqref{eq:coding_TVD_to_Euc} at least $tn$ terms are $\geq 4\beta^2$ (and could be $0$ in the others), therefore
\begin{align}
-\ln\left( 1-\frac12\|W_{x^n}-W_{{x'}^n}\|_1 \right) 
    &\geq \frac14 \sum_{i=1}^n \left| \sqrt{W_{x_i}}-\sqrt{W_{x_i'}}\right|_2^2\nonumber\\
    &\geq tn\beta^2.\label{eq_coding_epsilon_distance}
\end{align}

Thus we can choose $\epsilon = \exp(-t\beta^2 n)$. Using simple combinatorics, it is evident that the Hamming ball around any point in $\cX_0^n$ with radius $tn$ contains at most $\leq \binom{n}{tn} \Pi_\beta^{tn}$ elements. Therefore, any maximal code with a Hamming distance of $tn$ must have at least the following number of code words, which is the ratio of the total number of elements to the size of the Hamming ball (otherwise, the code could be extended):
\begin{equation}
  |\cC_t|\geq\Pi_\beta^{n(1-t)} 2^{-nH(t,1-t)}.
\end{equation}
The Gilbert-Varshamov bound \cite{Gilbert:combinatorics,Var:combinatorics} provides a similar (asymptotically equivalent) result, by constructing a linear code over the prime field $\mathbb{F}_p$, after reducing $\cX_0$ to the nearest smaller prime cardinality $p\geq \frac12|\cX_0|$ (Bertrand's postulate).
It only remains to bound the entropy difference $H(W_{x^n})-H(W_{x'^n})$ to apply Lemma \ref{lemma:hyp_test_error2}. As these entropies are in the interval $[0;n\log|\cY|]$, we can just partition the code into $S=\lceil n\log|\cY|\rceil$ parts $\cC_t^{(s)}$ ($s=1,\ldots,S$) in such a way that all $j\in\mathcal{C}_t^{(s)}$ have $H(W_{u_j})\in[s-1;s]$, meaning that any two $j,k\in\cC_s$ satisfy $|H(W_{u_j})-H(W_{u_k})|\leq 1$. We can now define $\cC$ as the largest $\cC_t^{(s)}$, which by the pigeonhole principle has $|\cC| \geq| \cC_t|/\lceil n\log|\cY|\rceil$ elements. 
We thus get the rate of the code as 
\begin{equation}
\label{eq:coding_rate_1}
\begin{split}
    R(n) 
    &=    \frac1n{\log|\cC|} 
    \geq \frac1n{\log|\cC_t|} 
           - \frac1n{\log \left\lceil n\log|\cY| \right\rceil} \\
    &\geq (1-t)\log\Pi_\beta - H(t,1-t) 
          -O\left(\frac{\log n}{n}\right).
\end{split}
\end{equation}
To this code, we finally apply Lemma \ref{lemma:hyp_test_error2} to bound the error of second kind, and obtain
\begin{equation}\begin{split}
\lambda_2
   &\leq 2\exp(-c\tau^2 n)
   + 3 \exp(2\tau n - t\beta^2 n)\\
   &\leq 5 \exp(-c\tau^2 n),
\end{split}\end{equation}
where the last inequality is true if $-c\tau^2 \geq 2\tau - t\beta^2$. Solving the quadratic constraint we obtain
\begin{equation}
  \tau \leq \frac{-2+\sqrt{4+4ct\beta^2}}{2c}
       =    \frac{\sqrt{1+ct\beta^2}-1}{c}.
\end{equation}
We may assume w.l.o.g.~$ct\beta^2 \leq 1$ (see the definition of $c$, the range of $t$ and the meaningful range of $\beta$; otherwise our rate lower bound is trivial), hence by the concavity of the square root the above condition will hold for 
$\tau = \left(\sqrt{2}-1\right)t\beta^2$. 
Notice that the two error exponents in the present construction are $E_1(n) \geq c\tau^2 - \frac1n$ and $E_2(n) \geq c\tau^2 - \frac3n$.
Then, they have the form claimed in the theorem if we choose $E(n) = \frac16 ct^2\beta^4$. To conclude, we solve for $\beta$ and substitute the solution $\beta = \sqrt[4]{\frac{6 E(n)}{ct^2}}$ into Equation~\eqref{eq:coding_rate_1}.
\end{proof}

\begin{remark}\label{remark:coding}
While the lower bound in Theorem \ref{thm:coding} is well suited for small error exponents 
, it becomes trivial for bigger exponents. Clearly, for $E(n)\geq 2ct^2/3$ the packing radius $\beta\geq\sqrt{2}$ is at least the maximum possible distance between any two elements in $\sqrt{\!\widetilde{\cX}}$, so the packing number can only be 1. This results on a non-positive (trivial) lower bound on the rate in \eqref{eq:coding_rate_final}.
The reason behind this behaviour and further remarks on the regime of large error exponents can be found in Section \ref{ssec:large_exponents}. 
However, for the main part of the discussion, where the small exponent regime is the relevant one as only there we will find super-linearity, this will not be a problem.
\end{remark}

\begin{corollary}
\label{cor:coding}
Let the constants $\eta>0$ and $0<t<1$, and small enough $E(n)>0$ be given. Then, for all sufficiently large $n$, there exists a DI code with $E_1(n),\, E_2(n) \geq E(n)-\frac3n$ and the following lower bound on the rate in terms of the lower Minkowski dimension:
\begin{equation}
 \label{eq:coding_corollary}
\begin{split}
   R(n) &\geq \frac{1-t}{4} \left[\underline{d}_M\left(\!\sqrt{\!\widetilde{\cX}}\right)-\eta\right] \log\frac{ct^2}{6E(n)} \\
   &\phantom{\geq.}
     - H(t,1-t)-O\!\left(\frac{\log n}{n}\right).
\end{split}
\end{equation}
\end{corollary}
\begin{proof}
By definition of the lower Minkowski dimension [Equation~\eqref{eq:lowerMinkowski}], for any $\eta>0$ there is an $E_0>0$ such that for all $E(n)\leq E_0$ we have
\begin{equation}
  \log\left[\Pi_{\sqrt[4]{\frac{6 E(n)}{ct^2}}} \left(\!\sqrt{\!\widetilde{\cX}}\right)\right]
    \geq (\underline{d}_M-\eta)\log \sqrt[4]{\frac{ct^2}{6E(n)}}.
\end{equation}
The proof is completed by direct substitution of the above expression into Equation~\eqref{eq:coding_rate_final} and taking the fourth square root out of the logarithm.
\end{proof}

While the bounds presented in this subsection might look a bit abstract, an important conclusion can be extracted easily. While in transmission the rate-reliability functions converge to the capacity value when the error exponent goes to 0, in deterministic identification it diverges! This behaviour was already hinted in the schematic Figure \ref{fig:T_vs_DI}, and it is clear looking at the bound in Corollary \ref{cor:coding}, where the term $\sim\log(1/E(n))\rightarrow\infty$ when $E(n)\rightarrow0$. Maybe a bit hidden, we of course have the same effect in Theorem \ref{thm:coding}, because the packing number of radius $\delta\propto\sqrt[4]{E(n)}\rightarrow0$ diverges to infinity. 
So, we have a lower bound that diverges for arbitrary small error exponents, meaning that we can construct a code that achieves infinite linear-scale rate. The correct way to interpret this result is that the rate is not defined in the correct scaling, a faster one is needed. Indeed, we will see in Section \ref{ssec:linear_to_super} that, in the regime of small error exponents, the rates scaling is linearithmic, and hence, if we define $R(n)=\frac{1}{n\log n}\log N$, the modified reliability function will converge to the (non-trivial) linearithmic capacity lower bound in \cite{CDBW:DI_classical} when $E(n)\rightarrow0$.

\subsection{Upper bound on the rate-reliability function (converse)}
\label{ssec:converse}
Following the structure of the previous section, we start presenting our general upper bound in Theorem \ref{thm:converse}, and analyse the regime of small error exponents in Corollary \ref{cor:converse}.
\begin{theorem}
\label{thm:converse}
For any DI code of block length $n$ with positive error exponents, $E_1(n),\, E_2(n) \geq E(n) > 0$, the rate is upper-bounded as follows:
\begin{equation}
  \label{eq:converse_rate}
  R(n) \leq \log\left[\Gamma_{\frac12{\sqrt{1-e^{-E(n)/2}}}}\left(\!\sqrt{\!\widetilde{\cX}}\right)\right].
\end{equation}
\end{theorem}

\begin{proof}
We start by noticing that, given two different codewords $u_j\neq u_k$ of any DI code, their output distributions have total variation distance bounded. Indeed, we know that there exist a subset $\cE_j$, the identification test, for which the probability to accept $u_j$ is $1-\lambda_1$, and $\lambda_2$ for any other $u_{k\neq j}$. Therefore, by definition of the total variation distance as a supremum over all possible output subsets [see Equation \eqref{eq:TVD}], it is clear that 
\begin{equation}\frac12\|W_{u_j}-W_{u_k}\|_1 \geq|W_{u_j}(\cE_j)-W_{u_k}(\cE_j)|\geq 1-\lambda_1-\lambda_2.\end{equation} 
Thus,
\begin{equation}
  1-\frac12\|W_{u_j}-W_{u_k}\|_1 \leq 2e^{-E(n)n}.
\end{equation}
Using the relations between fidelity and total variation distance in Equation~\eqref{eq:Classical_FvdG}, we get
\begin{equation}
\label{eq:converse_maximum_F}
    F(W_{u_j},W_{u_k})^2 \leq 1-\left(1-2e^{-E(n)n}\right)^2
    < 4e^{-E(n)n}.
\end{equation}
This is the maximum possible fidelity between the output distributions of any two code words of a good DI code. Now, our strategy consists in creating a covering such that elements in the same ball have a fidelity higher than the limit stipulated by Equation~\eqref{eq:converse_maximum_F}. In other words, we will ensure that elements in the same ball cannot be well-distinguishable code words, so that only one code word can be on each ball. Upper-bounding the number of code words (and therefore the rate) is then reduced to upper-bounding the cardinality of the covering. 

We start with an Euclidean $\frac{\delta}{2}$-covering $\cX_0$ of the square root output probability set $\sqrt{\!\widetilde{\cX}}$ which has $|\cX_0|=\Gamma_{\delta/2}\!\left(\!\sqrt{\!\widetilde{\cX}}\right) =: \Gamma_{\delta/2}$ elements,  meaning that for every $x\in\cX$ there is a $\xi\in\cX_0$ such that
\begin{equation}
  \left| \sqrt{W_x}-\sqrt{W_\xi} \right|_2 \leq \frac{\delta}{2}.
\end{equation}
Hence, our input set $\cX$ is partitioned according to the $\frac{\delta}{2}$-balls around each $\sqrt{W_\xi}$ with $\xi\in\cX_0$. That is, to each ball $\ell$ corresponds a set of the input sequences $\cX_\ell$ that generate only probability distributions in that ball. Meaning that  
for any two $x,x'\in\cX_\ell$ we have $|\sqrt{W_x}-\sqrt{W_{x'}}|_2 \leq \delta$. 
We know from Equation~\eqref{eq:fidelity->euclidean} that $1-F(W_x,W_{x'})^2\leq|\sqrt{W_x}-\sqrt{W_{x'}}|_2^2$. Then,
\begin{equation}\label{eq:converse_fidelity_step}
F(W_x,W_{x'})^2\geq1-\delta^2.
\end{equation}
The whole input space is partitioned as $\cX=\dot{\bigcup}_{\ell=1}^{|\cX_0|}\cX_\ell$.
In block length $n$, this gives rise to the partition $\cX^n = \dot{\bigcup}_{\ell^n} \cX_{\ell^n}$, where $\cX_{\ell^n} := \cX_{\ell_1}\times\dots\times\cX_{\ell_n}$ and $\ell^n\in[|\cX_0|]^n$ is the sequence that characterizes the ball for each letter. By virtue of the multiplicativity of the fidelity under tensor products and Equation~\eqref{eq:converse_fidelity_step}, we have for all $x^n,x'^n\in\cX_{\ell^n}$, 
\begin{equation}
  F(W_{x^n},W_{x'^n})^2 \geq (1-\delta^2)^n.
\end{equation}
If we now choose $1-\delta^2=e^{-E(n)/2}$, we find that as long as $e^{-E(n)n/2}\geq 4e^{-E(n)n}$ (i.e.~$E(n)n/2 \geq \ln 4$) the fidelity between elements of the same $\cX_{\ell^n}$ is bigger than the maximum possible fidelity that the DI code allows. In other words, there cannot be two code words in the same $\cX_{\ell^n}$. So we can upper bound the number of elements in our code by counting the cardinality of the index $[|\cX_0|]^n$: 
$N\leq |\cX_0|^n = \left(\Gamma_{\delta/2}\right)^n$ with the chosen value $\delta=\sqrt{1-e^{-E(n)/2}}$. We get
\begin{equation}
    R(n)=    \frac{1}{n}\log N 
        \leq \log\left[\Gamma_{\frac12{\sqrt{1-e^{-E(n)/2}}}}\left(\!\sqrt{\!\widetilde{\cX}}\right)\right],
\end{equation}
concluding the proof.
\end{proof}

\begin{corollary}
\label{cor:converse}
Let $\eta>0$ and a small enough $E(n)>0$ be given. Then, for all sufficiently large $n$, the rate of any DI code with error exponents $E_1(n),\,E_2(n) \geq E(n)$ can be upper-bounded in terms of the upper Minkowski dimension $\overline{d}_M$ as
\begin{equation}
  \label{eq:converse_corollary}
  R(n)\leq \frac12\left[\overline{d}_M\left(\!\sqrt{\!\widetilde{\cX}}\right)+\eta\right] \log\frac{8}{E(n)} + O\left( E(n) \right).
\end{equation}
\end{corollary}
\begin{proof}
By the definition of the upper Minkowski dimension, for all $\eta>0$ there exists an $E_0>0$ such that for all $E(n) \leq E_0$ it holds 
\begin{equation}\label{eq:plotted_upper_bound}
\log\left[\Gamma_{\frac12\sqrt{1-e^{-E(n)/2}}}\left(\!\sqrt{\!\widetilde{\cX}}\right)\right]
\!\leq\! \left(\overline{d}_M+\eta\right) 
\log\frac{2}{\sqrt{1-e^{-E(n)/2}}}.
\end{equation}
For $E(n)$ small (close to zero) we can expand the log as follows,
\begin{equation}
  \label{eq:Puiseux_series}
  \log \frac{2}{\sqrt{1-e^{-E(n)/2}}}
    \leq \frac12\log\frac{8}{E(n)} + O\left(E(n)\right).
\end{equation}
The proof is completed by substituting the above expressions into Equation~\eqref{eq:converse_rate}. 
\end{proof}

\subsection{Improved bounds for subsets of error exponents}
\label{ssec:improved}
In Corollaries \ref{cor:coding} and \ref{cor:converse} we have lower- and upper-bounded the rate-reliability function in terms of the lower and upper Minkowski dimensions, respectively. These bounds are universal in the sense that they hold for all error exponents $E$ smaller than some threshold $E_0=E_0(\eta)$ that decreases with the arbitrary parameter $\eta>0$. 
We might however be interested only in certain (infinitely many) sufficiently small values of $E$. To capture this, we consider subsets $\cE \subset (0;+\infty)$ that have $0$ as an accumulation point (meaning that $\cE$ contains sequences converging to $0$). 

Let us start with a subset $\cE_g$ such that
\begin{equation}
  \lim_{\cE_g \ni E\rightarrow 0} \frac{\log\left[\Pi_{\sqrt[4]{6E/ct^2}} \left(\!\sqrt{\!\widetilde{\cX}}\right)\right]}{-\log\sqrt[4]{6E/ct^2}}
  = \overline{d}_M\!\left(\!\sqrt{\!\widetilde{\cX}}\right).
\end{equation}
This is a ``good'' subset of error exponents, in the sense that it implies the convergence of $(\log\Pi_\delta)/(-\log\delta)$ to the best possible value (the limsup) as $\delta\rightarrow 0$. Furthermore, due to the properties of $\cE_g$, given any $\eta>0$ and its threshold $E_0>0$, there exist values $\cE_g \ni E \leq E_0$, and for those we have
\begin{equation}
  \label{eq:improvement_coding}
  \log\left[\Pi_{\sqrt[4]{\frac{6E}{ct^2}}}\left(\!\sqrt{\!\widetilde{\cX}}\right)\right]
  \geq \frac{1}{4}(\overline{d}_M-\eta)\log\frac{ct^2}{6E}.
\end{equation}
Substituting the above expression into Equation~\eqref{eq:coding_rate_final}, we obtain an improved achievability bound in Corollary \ref{cor:coding}, but now with the upper Minkowski dimension and under the added condition that the error exponents $E(n) \in \cE_g$.

We could similarly choose a ``bad'' subset $\cE_b\subset (0;+\infty)$ such that the covering numbers converge to the lower Minkowski dimension:
\begin{equation}
  \lim_{\cE_b \ni E \rightarrow 0} \frac{\log\left[\Gamma_{\frac12\sqrt{1-e^{-E/2}}}\left(\!\sqrt{\!\widetilde{\cX}}\right)\right]}{-\log\left(\frac12\sqrt{1-e^{-E/2}}\right)}
  = \underline{d}_M\!\left(\!\sqrt{\!\widetilde{\cX}}\right).
\end{equation}
As before, for all $\eta>0$ and its associated threshold $E_0$, there exist values $\cE_b \ni E \leq E_0$, for all of which
\begin{equation}
  \label{eq:improvement_converse}
  \log\left[\Gamma_{\frac12\sqrt{1-e^{-E/2}}}\left(\!\sqrt{\!\widetilde{\cX}}\right)\right]
   \leq \frac12(\underline{d}_M+\eta)\log\frac{8}{E}+O(E),
\end{equation}
where we have followed the last steps in the proof of Corollary \ref{cor:converse}, using again the expansion in Equation~\eqref{eq:Puiseux_series}. Plugging the above expression into Equation~\eqref{eq:converse_rate}, we obtain an improved upper bound on the rate-reliability function in Corollary \ref{cor:converse}, but now with the lower Minkowski dimension and under the added condition that the error exponents $E(n) \in \cE_b$.
These improved bounds for good and bad subsets of error exponents will be realted to the optimistic and pessimistic capacity approaches respectively in the next section, specifically in \ref{ssec:linear_to_super}.

\section{Discussion: constant vs vanishing exponents, dimension zero, and large exponent regime}
\label{sec:discussion}
In this section, we demonstrate how the new reliability framework developed in this study can yield the established results in \cite{CDBW:DI_classical}. 
We begin by deriving bounds for the rate-reliability functions in the asymptotic regime ($n \rightarrow \infty$), showing how these bounds allow us to recover the main established capacity results under some particular error conditions. Next, we consider the special case where the Minkowski dimension of the output probability set is zero---a scenario where Corollaries \ref{cor:coding} and \ref{cor:converse} do not fully capture the underlying behaviour meaningfully. We illustrate an approach to handle such cases with a couple of straightforward examples. Finally, in \ref{ssec:large_exponents}, we discuss the regime of large error exponents which, as we have already hinted in \ref{remark:coding}, is not bounded meaningfully through the presented Theorems. 

\subsection{From linear to linearithmic DI rates}
\label{ssec:linear_to_super}
The results from the preceding section tell us that the superlinear rates which, as discussed in the introduction, are a defining characteristic of general DI codes, are lost when imposing exponentially vanishing errors. Indeed, in the asymptotic limit $n\rightarrow\infty$, let $E = \min\{E_1,E_2,E_0(\eta)\} > 0$, then Corollary \ref{cor:converse} shows that for every $\eta>0$:
\begin{equation}
  \label{eq:rate-reliabiliyty-outerbound}
  \overline{R}(E_1,E_2) 
   \leq \frac12 \left(\overline{d}_M + \eta\right)
                \log\frac{8}{E} + O\left(E\right) 
   < \infty.
\end{equation}
This is because we can take the limit in Equation~\eqref{eq:converse_corollary} for every parameter $\eta>0$, as long as $E$ is below the cutoff exponent $E_0(\eta)$. 

Conversely, and similarly, for every $\eta>0$ and $0<t<1$, Corollary \ref{cor:coding} implies that in the asymptotic limit and for all $E_1,\,E_2 \leq E$:
\begin{equation}
  \label{eq:rate-reliabiliyty-innerbound}
  R(E_1,E_2) 
   \geq \frac{1}{4} (1\!-\!t)(\underline{d}_M-\eta) 
               \log\frac{ct^2}{6E} - H(t,1\!-\!t), 
\end{equation}
which is positive for sufficiently small $E>0$. Indeed, Equations \eqref{eq:rate-reliabiliyty-outerbound} and \eqref{eq:rate-reliabiliyty-innerbound} tell us that the linear-scale rate for DI is a bounded positive number under exponentially vanishing errors (the condition under which the previous Section \ref{sec:bounding} is developed) with $E>0$.

At the same time, however, we notice that the DI rates increase when we lower $E$. Actually, both Equations \eqref{eq:rate-reliabiliyty-outerbound} and \eqref{eq:rate-reliabiliyty-innerbound} diverge to infinity as $E\rightarrow 0$! Meaning that the rate scaling is faster than linear in this regime. As a matter of fact, we show next that the general linearithmic scaling of DI codes featured in \cite{CDBW:DI_classical} does reappear when we consider (slowly) vanishing error exponents.

Let $\lim_{n\rightarrow\infty} E_1(n) = \lim_{n\rightarrow\infty} E_2(n) = 0$. 
Evidently, we still need small errors $\lambda_1,\,\lambda_2 \ll 1$ as $n\rightarrow \infty$ to have a good code, so the $n E_{1,2}(n)$ must remain bounded away from $0$; more precisely, we need $E_1(n),\, E_2(n) \geq \omega(1/n)$ for the coding part and weak converse, and $E_1(n),\, E_2(n) \geq \Omega(1/n)$ for the strong converse. By direct substitution of these two ``settings'' into the results above we can recover the results from \cite{CDBW:DI_classical}. Before proving it, we point out that those results are bounds on the linearithmic capacity (the maximum asymptotic rate), which we can recover as follows in the present notation: 
\begin{equation}
  \!\!\dot{C}_\text{DI}(W) 
    =\!\! \sup_{\myfrac{E_i(n)\rightarrow 0,}{nE_i(n)\rightarrow\infty}} \!\!
      \left( \liminf_{n\rightarrow\infty} \frac{1}{\log n} \max_{\myfrac{\text{errors }\lambda_i\text{ s.t.}}{\log \lambda_i \leq -nE_i(n)}}\!\! R(n) \right),
\end{equation}
where the inner maximum is over all DI codes with errors of first and second kind bounded as $\lambda_{1} \leq 2^{-nE_{1}(n)}$ and $\lambda_{2} \leq 2^{-nE_{2}(n)}$, and the outer supremum is over functional dependencies of the exponents on $n$ going to $0$ sufficiently slowly: $\omega(1/n) \leq E_{1,2}(n) \leq o(1)$.

\begin{theorem}[{Cf.~\cite[Thm.~5.7]{CDBW:DI_classical}}]
\label{thm:recover_pessimistic}
The linearithmic DI capacity $\dot{C}_{\text{DI}}(W)$ for the channel $W:\cX\rightarrow\cY$ is bounded in terms of the lower Minkowski dimension of the set $\sqrt{\!\widetilde{\cX}}$ as follows (for slowly vanishing error exponents $\displaystyle\lim_{n\rightarrow\infty} E_i(n) = 0$):
\begin{equation}
  \frac14\underline{d}_M \!\left(\!\sqrt{\!\widetilde{\cX}}\right)
    \leq \dot{C}_{\text{DI}}(W)
    \leq \frac12\underline{d}_M\!\left(\!\sqrt{\!\widetilde{\cX}}\right).
\end{equation}
\end{theorem}
\begin{proof}
We start with the coding part taking $E_1(n) = E_2(n) = \frac{\log n}{n}$ in Equation~\eqref{eq:coding_corollary}, from Corollary \ref{cor:coding}:
\begin{equation}\begin{split}
R(n) 
    &\geq \frac{1-t}{4} \left[\underline{d}_M \left(\!\sqrt{\!\widetilde{\cX}}\right)-\eta\right] \log\frac{ct^2n}{6\log n}\\
    &\phantom{\geq.}-H(t,1\!-\!t)-O\!\left(\frac{\log n}{n}\right).
\end{split}\end{equation}
We can maximize this rate choosing the parameter $t=t(n)>0$ going to zero slowly enough for increasing $n$, such that $\lim_{n\rightarrow\infty} \frac{\log(t^2n/\log n)}{\log n} = 1$ (we can for example choose $t=1/\log n$), and $\eta=\eta(n)>0$ also going to zero as $n\rightarrow\infty$. Then, when taking the limit of large $n$ we find
\begin{equation}
  \dot{C}_{\text{DI}}(W)
    \geq \liminf_{n\rightarrow\infty} \frac{1}{\log n} R(n)
    \geq \frac14 \underline{d}_M\left(\!\sqrt{\!\widetilde{\cX}}\right),
\end{equation}
completing the direct part.

Similarly, for the upper bound we need to take $\cE_b\ni E_1(n) = E_2(n) \geq \frac{C}{n}$ in Equation~\eqref{eq:improvement_converse}, the (pessimistic) upper bound for bad sequences of errors. Choosing also $\eta=\eta(n)$ going slowly enough towards $0$ for increasing $n$, such that the error exponents, which are upper bounded by the threshold $E_0\!\left(\eta(n)\right)$, can also go to $0$ slowly, we find
\begin{equation}\begin{split}
  \dot{C}_{\text{DI}}(W)
    &\leq \liminf_{n\rightarrow\infty} \frac{1}{2\log n} \underline{d}_M\left(\!\sqrt{\!\widetilde{\cX}}\right) \log\frac{8n}{C} + O\left(\frac1n\right)\\
    &= \frac12 \underline{d}_M\left(\!\sqrt{\!\widetilde{\cX}}\right),
\end{split}\end{equation}
and we are done.
\end{proof}

The capacity $\dot{C}_{\text{DI}}$ has to be approach by all sufficiently large values of $n$, i.e. we take the lowest convergence value of the limit (the inferior limit), this is why it is commonly referred to as a \textit{pessimistic capacity}. It might be possible to find a subsequence of block lengths $n_k\rightarrow\infty$ for which the rate converges to a higher number. This leads to the definition of the \textit{optimistic capacity}, which provides information about the best achievable rates (though only for specific block lengths) through the superior limit:
\begin{equation}
 \! \!\dot{C}_{\text{DI}}^\text{opt}(W) 
  = \!\!\!\sup_{\myfrac{E_i(n)\rightarrow 0,}{nE_i(n)\rightarrow\infty}} \!\!\!
      \left( \limsup_{n\rightarrow\infty} \frac{1}{\log n} \max_{\myfrac{\text{errors }\lambda_i\text{ s.t.}}{\log \lambda_i \leq -nE_i(n)}}\!\!\! R(n) \right),
\end{equation}
with the maximum inside and the supremum outside ranging over the same objects as in the formula for the pessimistic capacity. 

We refer the reader to \cite[Sect.~7.1]{CDBW:DI_classical} for an extensive discussion on the origins, comparison, and interpretation of the optimistic and pessimistic capacities, or else \cite{Ahlswede2006} which inspired it.
We can recover the general superexponential optimistic capacity bounds in \cite{CDBW:DI_classical} with similar methods as in the previous theorem.

\begin{theorem}[{Cf.~\cite[Thm.~5.9]{CDBW:DI_classical}}]\label{thm:recover_optimistic}
The optimistic linearithmic DI capacity $\dot{C}_{\text{DI}}^{\text{opt}}(W)$ of a channel $W:\cX\rightarrow\cY$ is bounded in terms of the upper Minkowski dimension of $\sqrt{\!\widetilde{\cX}}$ as follows (for slowly vanishing error exponents $\displaystyle\lim_{n\rightarrow\infty} E_i(n) = 0$):
\begin{equation}
  \frac14\overline{d}_M\! \left(\!\sqrt{\!\widetilde{\cX}}\right)
    \leq \dot{C}_{\text{DI}}^{\text{opt}}(W) 
    \leq \frac12\overline{d}_M\!\left(\!\sqrt{\!\widetilde{\cX}}\right).
\end{equation}
\end{theorem}
\begin{proof}
We follow the same steps as in the previous proof. For the direct part we need to take $\cE_g\ni E_1(n),\, E_2(n) \geq \frac{C}{n}$ in the (optimistic) lower bound for good sequences of error exponents [Equation~\eqref{eq:improvement_coding}], and similarly choose $t=t(n)$ and $\eta=\eta(n)$ going to $0$ simultaneously and sufficiently slowly as $n\rightarrow\infty$. Then,
\begin{align}
  \dot{C}_{\text{DI}}^{\text{opt}}(W)
    &\geq \limsup_{n\rightarrow\infty} \frac{1}{4\log n} \left[\overline{d}_M\!\left(\!\sqrt{\!\widetilde{\cX}}\right)-\eta(n)\right]\log\frac{ct^2n}{6C}\nonumber\\
    &= \frac14\overline{d}_M\! \left(\!\sqrt{\!\widetilde{\cX}}\right).
\end{align} 

For the converse can we take $E_1(n) = E_2(n) = \frac{\log n}{n}$ in Equation~\eqref{eq:converse_corollary} and $\eta=\eta(n)$ going to zero slowly enough, and find
\begin{align}
  \dot{C}_{\text{DI}}^{\text{opt}}(W)
    &\leq \limsup_{n\rightarrow\infty} \frac{1}{2\log n} \left[\overline{d}_M\!\left(\!\sqrt{\!\widetilde{\cX}}\right)+\eta(n)\right] \log\frac{8n}{\log n}\nonumber\\
    &\phantom{\leq.}+ O\left( \frac1n \right) = \frac12\overline{d}_M\! \left(\!\sqrt{\!\widetilde{\cX}}\right),
\end{align}
completing the proof.
\end{proof}

In conclusion, while for exponentially fast vanishing errors the superlinear nature of DI is lost, when imposing a very slow decrease of the errors in block length we can recover the linearithmic capacity bounds in \cite{CDBW:DI_classical} from the rate-reliability bounds derived in Section \ref{sec:bounding}. This shows that, despite the slightly superlinear nature of the optimum rates being a general feature of DI, it is also very sensitive to the particular error regimes. 
The convergence of the rate bounds towards the capacity is also very slow, as shown in Figure~\ref{Fig:finite_n}.

\begin{figure}[ht]
    \centering
    \includegraphics[width=0.95\linewidth]{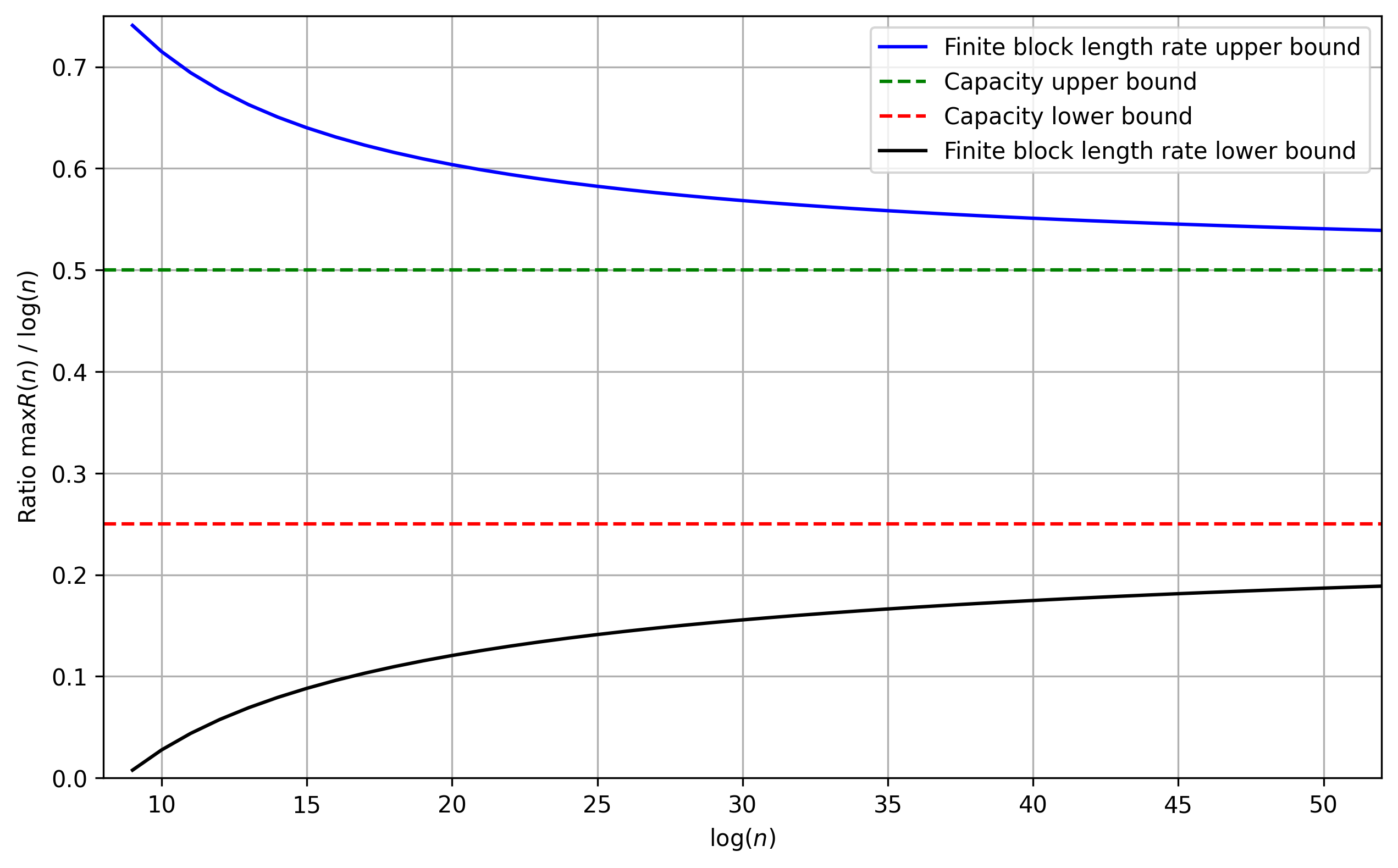}
    \caption{\small Tendency of the upper and lower rate bounds towards the capacity upper and lower bounds for increasing block length, in an example case with $d_M=1$ and $E_1=E_2=E(n)= 1/n$. The upper bound (blue line) is a plot of Eq.~\eqref{eq:plotted_upper_bound} with $\eta(n)=1/(\log n)$, and the lower bound (black line) a plot of Eq.~\eqref{eq:coding_corollary} with $t(n)^2={3}/{(c\log n)}$ and $\eta(n)=1/n$, following the conditions described in the proofs of Theorems \ref{thm:recover_pessimistic} and \ref{thm:recover_optimistic} above.}
    \label{Fig:finite_n}
\end{figure}


\subsection{Zero-dimensional output probability sets}
\label{ssec:d=0}

In Corollaries \ref{cor:coding} and \ref{cor:converse}, and in Section~\ref{ssec:improved}, we have demonstrated how to bound the rate $R(n)$ of a code with small error exponents using the lower and upper Minkowski dimensions featured in the main pessimistic and optimistic capacity results of \cite{CDBW:DI_classical}. 
This could cause some confusion on how to approach channels for which the output probability set (and, equivalently, its square-root) have Minkowski dimension $d_M\!\left(\!\sqrt{\!\widetilde{\cX}}\right)=0$, particularly in the regime of small error exponents, where the aforementioned corollaries apply. Indeed, in such a case, the lower bound given in Corollary \ref{cor:coding} becomes trivial, and the upper bound only tells us that $R(n) \leq o\left(\log\frac{1}{E(n)}\right)$ for small $E(n)\leq E_1(n),\,E_2(n)$. 
However, we can obtain much more meaningful information if we go back to Theorems \ref{thm:coding} and \ref{thm:converse} (for the lower and upper bounds respectively) and impose the regime of small errors only at an appropriate point.

Let us demonstrate this through two examples. To state the first, we recall the definition of the \emph{Bernoulli channel} $B:[0;1] \rightarrow \{0,1\}$ on input $x\in[0;1]$ (a real number), which outputs a binary variable distributed according to the Bernoulli distribution $B_x$ with parameter $x$:
\begin{equation}
    \label{eq:BernoulliChannel}
     B_x(y) = xy + (1-x)(1-y)
      = \begin{cases}
          x   & \text{for } y=1, \\
          1-x & \text{for } y=0.
        \end{cases}
\end{equation}

\begin{example}[{Cf.~\cite[Subsect.~5.4]{CDBW:DI_classical}}]
\label{ex:Bernoulli}
For a real number $a>1$, let $\cX_a:=\{0\}\cup\{a^{-k}:0\leq k\in\mathds{N}_0\}$ be a set of numbers in the interval $[0,1]$. 
Then, the maximum rate $R(n)$ of a DI code for $B\vert_{\cX_a}$, the Bernoulli channel restricted to inputs from $\cX_a$, for small and equal error exponents $E(n):=E_1(n)=E_2(n)>0$, is bounded as follows:
\begin{equation}
  \log\log \frac{1}{E(n)} - O(1)
    \leq \max R(n)
    \leq \log\log\frac{1}{E(n)} + O(1).
\end{equation}
The capacity of this channel in the suitably defined $n\log\log n$ scale is thus $\mathring{C}_\text{DI}(\cB|_{\cX_a})=1$.
\end{example}

\begin{proof}
As the set of all probability distributions with binary output $\cP(\{0,1\})$ (with total variation distance) is isometric to the interval $[0;1]$ (with the usual metric on real numbers), the covering (or packing) of the output probability set of a channel with binary output, like $B|_{\cX_a}$, corresponds to the covering (or packing) of its input, which is $\cX_a$. 
It is not hard to see that
\begin{equation}
  \label{eq:example_covering_bounds}
  \log_a\frac{a-1}{3\epsilon} 
    \leq \Gamma_{\epsilon}(\cX_a) 
    \leq \log_a\frac{a^2}{2\epsilon}.
\end{equation}
Indeed, the upper bound results from placing an interval of length $2\epsilon$ at $[0;2\epsilon]$, which covers all $k$ such that $a^{-k}\leq 2\epsilon$, plus a separate interval for each smaller $k$. The lower bound comes from the realization that if the gap between two consecutive points in $\cX_a$ is bigger than $2\epsilon$, then they cannot be covered by the same interval. Say, $2\epsilon < 3\epsilon \leq a^{1-k}-a^{-k} = (a-1)a^{-k}$, which is true for $k\leq \log\frac{a-1}{3\epsilon}$. This means that each such $a^{-k}$ needs its own interval. 

To prove the lower bound, we start from Theorem \ref{thm:coding}, we use the covering/packing relation $\Pi_{\epsilon}(F)\geq\Gamma_{2\epsilon}(F)$, and apply the lower bound in Equation~\eqref{eq:example_covering_bounds}, noting that in this example $\sqrt{\cX_a}=\cX_{\sqrt{a}}$.
After some basic algebra we get:
\begin{equation}\begin{split}
  R(n) 
   &\geq (1-t)\log\log_a \frac{t\sqrt{c}(\sqrt{a}-1)^2}{36\sqrt{6E(n)}}\\ 
   &\phantom{\geq.}- H(t,1-t) - O\left(\frac{\log n}{n}\right).
\end{split}\end{equation}
To optimise the right-hand side, choose $t=\sqrt[4]{E(n)}$ such that
\begin{equation}\begin{split}
\!\!R(n) &\geq \left(1-\sqrt[4]{E(n)}\right) \log \log_a \frac{O(1)}{\sqrt[4]{E(n)}}\\
&\phantom{==}- H\left(\sqrt[4]{E(n)},1-\sqrt[4]{E(n)}\right) -O\left(\frac{\log n}{n}\right)\\
&=\log\log\frac{1}{E(n)} - O(1).
\end{split}\end{equation}

For the upper bound we start from Theorem \ref{thm:converse} and use the right-hand inequality in Equation~\eqref{eq:example_covering_bounds} and that $\sqrt{\cX_a}=\cX_{\sqrt{a}}$:
\begin{equation}\begin{split}
  R(n) &\leq \log\left[\Gamma_{\frac12{\sqrt{1-e^{-E(n)/2}}}}\left(\!\sqrt{\!\widetilde{\cX}}\right)\right]\\
  &\leq\log \log_{\sqrt{a}}\left(\frac{a}{\sqrt{1-e^{-E(n)/2}}}\right).
\end{split}\end{equation}
For small error exponents $E(n)>0$ we can expand the expression above as follows:
\begin{equation}\begin{split}
  R(n)
  &\leq \log\log_a \frac{2a^2}{E(n)} + O\left(\frac{E(n)}{-\log E(n)}\right)\\
  &=\log\log\frac{1}{E(n)} +O(1).
\end{split}\end{equation}

Considering the bounds above and the arguments in Section \ref{ssec:linear_to_super}, where we took $E(n)=\omega(1/n)$ for the lower bound and $E(n)=\Omega(1/n)$ for the upper, we conclude that the DI capacity in this case ought be defined at $n\log\log n$ scale (cf.~\cite[Subsect.~5.4]{CDBW:DI_classical}), and that it evaluates to 
\begin{equation}
  \mathring{C}_\text{DI}(\cB|_{\cX_a}) 
    \!:=\!\!\!\! \sup_{\myfrac{E_i(n)\rightarrow 0,}{nE_i(n)\rightarrow\infty}} \!\!\!\!
      \left( \liminf_{n\rightarrow\infty}\frac{1}{\log\log n} \!\!\!\!\!\!\max_{\myfrac{\text{errors }\lambda_i\text{ s.t.}}{\phantom{==}\log \lambda_i \leq -nE_i(n)}} \!\!\!\!\!\!\!\! R(n) \right)\!=\! 1,
\end{equation}
completing the proof.
\end{proof}

\begin{example}[Cf.~\cite{SPBD:DI_power}]
\label{ex:DMC}
Given small enough $E(n)>0$, for all sufficiently large $n$ there exists a DI code with $E_1(n),\,E_2(n) \geq E(n)-\frac3n$ for a DMC $W:\cX\rightarrow\cY$ with $|W(\cX)|$ pairwise different output probability distributions (equivalently: $N_{\text{row}}=|W(\cX)|$ distinct rows of the channel stochastic matrix) and rate lower-bounded
\begin{equation}
  \label{eq:DMC-direct}
  R(n)\! \geq\! \log |W(\cX)| - O\left(-\sqrt{E(n)}\log E(n)\right) - O\left(\frac{\log n}{n}\right). 
\end{equation}
Conversely, for any DI code with $E_1(n),\,E_2(n) \geq E(n) > 0$ and sufficiently large $n$,
\begin{equation}\label{eq:DMC-converse}
  R(n) \leq \log |W(\cX)| - O\Bigl(-E(n)\log E(n)\Bigr).
\end{equation}
The linear capacity results for the DMC from \cite{AC:DI,SPBD:DI_power} are recovered from these two bounds in the limit of vanishing exponents: $C_{\text{DI}}(W) = \log |W(\cX)|$.
\end{example}

\begin{proof}
For the direct part [Equation~\eqref{eq:DMC-direct}] we can reuse the proof of Theorem \ref{thm:coding} with slight modifications. There, we started by defining a packing on $\sqrt{\!\cX}$ of radius $\beta$, and worked after with the discrete set of $|\cX_0|$ packing elements. Since now we already have a finite set of output probability distributions $\widetilde{\cX} = W(\cX)$, we let 
\begin{equation}
2\beta 
:= \min_{x\neq x'\in\cX}\left|\sqrt{W_x}-\sqrt{W_{x'}}\right|_2,
\end{equation} 
the minimum Euclidean distance between square roots of output probability distributions, and the rest of the proof is recycled. We will create a maximum-size code in block length $n$ with minimum Hamming distance between different code words $d_H(x^n,x'^n)\geq tn$ for $0<t<1$, and partition the code into regions of similar entropy. Following the arguments, we can lower-bound the rate as [cf.~Equation~\eqref{eq:coding_rate_1}]:
\begin{equation}
 R(n) \geq (1-t)\log|W(\cX)| - H(t,1-t) - O\left(\frac{\log n}{n}\right),
\end{equation}
and the error exponents are as claimed in the theorem by requiring $E(n)=\frac16 ct^2\beta^4$. Solving for $t$ and substituting in the equation above we get that, for small error exponents:
\begin{equation}
  R(n) \geq \log|W(\cX)| - O\left(-\sqrt{E(n)}\log E(n)\right) - O\left(\frac{\log n}{n}\right).
\end{equation}

For the converse [Equation~\eqref{eq:DMC-converse}] we start by noticing that Equation~\eqref{eq:converse_maximum_F} translates to a minimum Hamming distance between code words. 
Namely, letting $\frac{1}{\alpha} := \displaystyle \max_{x\neq x'\in\cX} F(W_x,W_{x'}) < 1$ be the maximum fidelity between different output probability distributions, we have for two distinct code words $u_j = x^n$ and $u_k = {x'}^n$:
\begin{equation}
  \frac{1}{\alpha^{d_H(u_j,u_k)}}\leq 4e^{-nE(n)},
\end{equation}
which implies
\(
  d_H(u_j,u_k) \geq nE(n)\log_\alpha e - \log_\alpha 4.
\)
With the sphere packing (aka Hamming) bound and in the regime of small error exponents we obtain
\begin{equation}
  R(n) \leq \log |W(\cX)| - O\left(-E(n)\log E(n)\right).
\end{equation}

The only thing left to show is that in the asymptotic limit, the rate bounds above yield the linear capacity results from \cite{SPBD:DI_power}. Repeating the argument in Section \ref{ssec:linear_to_super} we take $E(n)\geq\Omega(1/n)$ for the upper bound and $E(n)\geq\omega(1/n)$ for the lower, which translates into:
\begin{equation}\begin{split}
\max R(n) &\geq \log |W(\cX)| - O\left(\frac{1}{\sqrt{n}}\log n\right)\\
\max R(n) &\leq \log |W(\cX)| - O\left(\frac1n\log n\right).
\end{split}\end{equation}
As both the $O(\,\cdot\,)$ terms go to zero in the asymptotic limit, the linear capacity is given by
\begin{equation}
C_\text{DI}(W) 
    = \!\!\!\sup_{\myfrac{E_i(n)\rightarrow 0,}{nE_i(n)\rightarrow\infty}} \!
      \left( \liminf_{n\rightarrow\infty} \!\!\!\!\!\!\!\max_{\myfrac{\text{errors }\lambda_i\text{ s.t.}}{\phantom{==}\log \lambda_i \leq -nE_i(n)}}\!\!\!\!\!\!\!\! R(n) \right)=\log |W(\cX)|,
\end{equation}
and the proof is completed.
\end{proof}

Example \ref{ex:DMC} can be generalized to DMCs with an input power constraint. Given a cost function $\phi : \cX \rightarrow \RR_{\geq 0}$ (power per symbol), an input power constraint $A$ is an upper bound on the average power that admissible code words have to obey:
$\phi^n(x^n)=\frac1n\sum_{t=1}^n \phi(x_t)\leq A$.
From \cite{SPBD:DI_power} we know that the DI capacity of the DMC $W$ under the power constraint $A$ is 
\begin{equation}
  C_\text{DI}(W;\phi,A) = \max_{p_X\,:\,\EE \phi(X) \leq A} H(X),
\end{equation}
after purging the channel of duplicate input symbols: if $W_x=W_{x'}$, and say $\phi(x') \geq \phi(x)$, remove the input with the larger cost ($x'$) from $\cX$ and keep the other one ($x$); in case of a tie decide arbitrarily. 
Example \ref{ex:DMC} above is recovered when the power constraint is irrelevant (if $A \geq \max_{x\in\cX} \phi(x)$), in which case the capacity becomes $\max_{p_X\in\cP(\cX)} H(X) = \log|\cX| = \log|W(\cX)|$. With the power constraint active, and for positive error exponents, we obtain lower and upper bounds on $R(n)$ as given in Equations~\eqref{eq:DMC-direct} and \eqref{eq:DMC-converse}, with $\log|W(\cX)|$ replaced by $\max_{p_X\,:\,\EE \phi(X) \leq A} H(X)$.

Despite not having a closed formula for the general case with Minkowski dimension zero, we show how to tackle such cases through the manipulation of Theorems \ref{thm:coding} and \ref{thm:converse}. It is possible that calculating the needed covering and packing numbers is difficult in some cases, but regardless, it is clear that we will find the rate-reliability function bounded in terms of some (sublinearly growing) function of $\log\frac{1}{E}$.

\subsection{Large error exponent regime}
\label{ssec:large_exponents}
The bounds in Section \ref{sec:bounding} have permitted an extensive analysis of the trade-off between rate and reliability in the regime of small error exponents. But, as already mentioned in Remark \ref{remark:coding}, our lower bound on the rate [Equation \eqref{eq:coding_rate_final}] does not behave well for bigger error exponents. In this section we will explain the reason why this happens and show that a bypass to the problem is not at all trivial. After that, we include a discussion on what do we know about the rate-reliability functions in the regime of large error exponents.

Let us start by understanding when and why our analysis is limited to the regime of small error exponents. We have to go back to the proof of Theorem \ref{thm:coding} where, in the second inequality of Equation \eqref{eq:coding_TVD_to_Euc}, we use the bound $-\ln F^2\geq 1-F^2$, with $F$ the fidelity between two output distributions. Notice that this bound is only a good approximation when the letterwise fidelity is close to 1. That is, when the probability distributions are similar, which can be related to a finer packing, so smaller error exponents. Conversely, large error exponents are related to wider packings, which result on fidelity values getting closer to $0$, making the bound $-\ln F^2\geq 1-F^2$ grossly inaccurate. So, when applying this bound, we limit our analysis to small error exponents. 

The reader might wonder now why it is at all necessary to apply such a problematic bound. Well, in order to apply our proof idea, we need a metric so that a packing can be defined, but the negative logarithm of the fidelity [second line in Equation \eqref{eq:coding_rate_final}] is not a well-defined distance, as it does not satisfy the triangle inequality. This is why we use the bound $-\ln F^2\geq 1-F^2$ together with the lower bound in Equation \eqref{eq:fidelity->euclidean}, conveniently obtaining an Euclidean metric. As we have already discussed, this step is a good approximation for the small exponent regime, which provides the framework for the superlinearity study and allows the recovery of the capacity results from \cite{CDBW:DI_classical} (which were the initial objectives of this work). However, it is increasingly inexact for bigger error exponents making the bound inadequate in the regime of large error exponents.

The lack of a lower bound able to well-approximate $-\ln F^2$ in the range of small fidelities (for which $-\ln F^2$ grows to infinity) and from which a genuine distance can be extracted, suggests that our proof will probably not be trivially generalised to the large error exponent regime. Different ideas are needed. However, we can make some comments on how the rate-reliability function behaves in the large exponent regime, as follows.

To start, we can upper-bound the rate for DI over any channel $W$ meaningfully for arbitrary error exponents starting from a transmission code for $W$. It is clear that any transmission code without randomness on the encoder and error probability $\lambda=2^{-nE_\text{T}}$ is also a DI code with $2^{-nE_\text{DI}}=\lambda_1=\lambda_2=\lambda$. This is because to identify any message $m$, we could use the transmission code to read the decoded message $\hat{m}$, and decide whether they are the same $m\stackrel{?}{=}\hat{m}$, hence the DI errors and rate would then be equal to the error and the rate of the transmission code used. Thus, the rate-reliability function for transmission is a lower bound on the rate-reliability function for DI: $E_\text{DI}(R) \geq E_\text{T}(R)$. Furthermore, by using a sequential decoding of the $N=2^{nR}$ transmission messages using an $(n,N,\lambda,\lambda)$-DI code (decoding a message by trying to identify each possible output), we can bound
\begin{equation}
  2^{-nE_\text{T}} 
   = \lambda
   \leq 2^{nR}2^{-nE_\text{DI}}
   =2^{-n(E_\text{DI}-R)},
\end{equation}
which directly implies the following upper bound on the rate-reliability function for DI: $E_\text{DI}(R) \leq E_\text{T}(R)+R$.

\section{Asymmetric error regimes}
\label{sec:steins}
All our previous results study the performance of DI codes with errors vanishing exponentially fast: $\lambda_1=2^{-nE_1}$ and $\lambda_2=2^{-nE_2}$. The error exponents $E_1,E_2>0$ are taken to be larger than some threshold $E_1,E_2>E$, and this $E$ is actually the defining value of the different bounds. Therefore, the codes analysed can only produce error-symmetric asymptotic achievable triples $(R,E,E)$. 

This begs the natural question: what happens to the rate when the errors are asymmetric? Can we achieve superlinear codes when only one error vanishes slowly? Here we analyse the following maximally asymmetric settings:
\begin{itemize}
    \item \textbf{Stein regime:} $\lambda_1<1$ and $\lambda_2=2^{-nE_2}$,
    \item \textbf{Sanov regime:} $\lambda_1=2^{-nE_1}$ and $\lambda_2<1$.
\end{itemize} 
These two regimes mirror two well-known classical asymmetric hypothesis testing settings. In that context, the goal is to distinguish between two probability distributions (hypotheses) $P$ and $Q$, defined on the same sample space $\cY$, using $n$ independent realizations of one of them. Two errors may occur: accepting $P$ (the null hypothesis) when the samples were produced by $Q$, and the converse case where $Q$ is accepted (the alternative hypothesis) while the samples come from $P$ \cite{CK:book2011}.  

\begin{remark}
In the settings described above, we refer to the two types of errors using the parameters \( \lambda_1 \) and \( \lambda_2 \), defined explicitly in the identification context. While it may be tempting to describe these errors as ``false positives'' or ``false negatives,'' such terminology can lead to confusion, as their interpretation is not consistent across domains.

For example, in classical hypothesis testing, the null hypothesis often models a healthy patient, and the alternative an illness. A ``positive'' result means accepting the alternative, so a type I error is interpreted as a false positive. In identification, on the other hand, a ``positive'' outcome naturally refers to accepting the message under test. In this case, the first type of error (failing to identify the transmitted message) becomes a false negative, and the second (accepting an incorrect message) becomes a false positive. This reversal illustrates how informal labels, which are widely used in the literature, can obscure rather than clarify. 
\end{remark}

In the Chernoff-Stein Lemma setting, the probability $\alpha_n$ of accepting the alternative hypothesis when the null is true is required to be upper bounded for all $n$ by $\epsilon\in(0,1)$; then, the probability $\beta_n$ of accepting the null hypothesis when the alternative is true decays exponentially with $n$, and the best achievable exponent is given by the Kullback-Leibler divergence $D(P\|Q):=\sum_{y\in\cY}P(y)\log[P(y)/Q(y)]$ of the two distributions. Indeed,
\begin{equation}
\lim_{n\rightarrow\infty}-\frac{1}{n}\log\beta_n=D(P\|Q).
\end{equation}

In the finite block length regime, a more refined quantity is used to characterize the tradeoff between the errors: the hypothesis testing relative entropy
\begin{equation}
  D_h^\epsilon(P\|Q) 
    := -\log\inf_{\cS\subset\cY\atop P(\cS^c)\leq\epsilon} Q(\cS),
\end{equation}
where $\cS$ is a decision region. This object converges to the KL divergence for increasing $n$: 
$\lim_{n\rightarrow\infty}\frac{1}{n}D_h^\epsilon(P^{\ox n}\|Q^{\ox n})=D(P\|Q)$ for any $\epsilon\in(0,1)$, connecting with the optimum exponent in the asymptotic regime. 

Clearly, if we take a good DI code in the Stein regime, the identification test $\cE_j$ satisfies $W_{u_j}(\cE_j^c) \leq \lambda_1$ and $W_{u_k}(\cE_j) \leq \lambda_2$ (for $k\neq j$). Hence we get immediately, by definition,
\begin{equation}
  D_h^{\lambda_1}(W_{u_j}\|W_{u_k}) \geq 
  -\log\lambda_2 = nE_2.
\end{equation}
Similarly, by taking $\cE_k$ as the test, we get for a good DI code in the Sanov regime that $D_h^{\lambda_2}(W_{u_j}\|W_{u_k}) \geq nE_1$. In other words, all good DI codes in the Stein (or Sanov) regime have output distributions with a minimum pairwise hypothesis testing relative entropy linear in $n$. We will use this property to prove converses for general channels in the Stein and Sanov error settings.

We are motivated to transform the entropies above into R\'enyi relative entropies, which are generally easier to handle, and above all additive. Given a parameter $\alpha\in(0,\infty)\setminus\{1\}$ they are defined as
\begin{equation}\label{eq:Renyi_def}
D_\alpha(P\|Q)=\frac{1}{\alpha-1}\log\sum_{y\in\cY}P(y)^\alpha Q(y)^{1-\alpha},
\end{equation}
and connected to the hypothesis testing relative entropies, for both $\lambda_1$ and $\lambda_2$ and $\alpha>1$ to be fixed later, through the following bound \cite{CMW16:Dh-to-D_alpha}:
\begin{equation}
\!\!D_h^{\epsilon}(W_{u_j}\|W_{u_k})\!\leq\! D_\alpha (W_{u_j}\|W_{u_k})+\frac{\alpha}{1-\alpha}\log\left(\!\frac{1}{1-\epsilon}\!\right).
\end{equation}
As the R\'enyi relative entropies are additive under product distributions, we will be able to transform terms referring to code words into sums of letter-wise objects. Indeed, given two codewords $u_j=x_1\dots x_n\in\cX^n$ and $u_k=x'_1\dots x'_n\in\cX^n$,
\(D_\alpha(W_{u_j}\|W_{u_k})=\sum_{i=1}^n D_\alpha(W_{x_i}\|W_{x'_i}).\) 

For our converse proof, we want to use the linear lower bounds of $D_h^{\lambda_{1,2}}$ described above, which in turn lower bound the sum of letter-wise R\'enyi entropies as follows:
\begin{equation}
\label{eq:Renyi_converse_bound}
\sum_{i=1}^n D_\alpha(W_{x_i}\|W_{x'_i})\geq nE_{1,2}-\frac{\alpha}{1-\alpha}\log\left(\frac{1}{1-\lambda_{2,1}}\right).
\end{equation}

Next, we create a partition of the code such that, no two different distributions in each subgroup satisfy the inequality above, meaning that each subset can only contain a single codeword. The task is therefore reduced to counting the number of subsets in the partition.

However, looking at Equation~\eqref{eq:Renyi_def}, we see that the sum can diverge if there are small values of $W_{x'_i}(y)$ (as this term has a negative exponent in the regime $\alpha>1$). In this case, the lower bound we want to use in Eq.~\eqref{eq:Renyi_converse_bound} becomes trivial. We overcome this issue by 
defining the following partition:
let $\widetilde{\cX}_\ell=W(\cX_\ell)$ be subsets of the output probability space $W(\cX)=:\widetilde{\cX}=\dot{\bigcup}_{\ell=1}^L\widetilde{\cX}_\ell$ such that for all $\ell\in[L]$, $x,x'\in\cX_\ell$, and $y\in\cY$:
\begin{equation}\label{eq:partition}
\frac{1}{1+\delta}W_{x'}(y)\leq W_x(y)\leq (1+\delta) W_{x'}(y),
\end{equation}
with $\delta>0$. With this, if for some $\ell\in[L]$ we have a problematic $W_{x'}(y)=0$, then all the other $x\in\cX_\ell$ are also $W_x(y)=0$, so the sum in Equation~\eqref{eq:Renyi_def} will no longer diverge for distributions in the same partition (it will just be 0). In consequence, the sum in the right hand side of Eq.~\eqref{eq:Renyi_converse_bound} will not diverge either. In general, for all $x_i,x_i'\in\cX_\ell$
\begin{equation}
\begin{split}
D_\alpha(W_{x_i}\|W_{x'_i})&=\frac{1}{\alpha-1}\log\sum_{y\in\cY}W_{x'_i}(y)\left(\frac{W_{x_i}(y)}{W_{x'_i}(y)}\right)^\alpha\\
&\leq \frac{1}{\alpha-1}\log\sum_{y\in\cY}W_{x'_i}(y)\left(1+\delta\right)^\alpha\\
&=\frac{\alpha}{\alpha-1}\log(1+\delta)\stackrel{!}{\leq}\frac{E_{1,2}}{2},
\end{split}
\end{equation}
where the first inequality follows from the upper bound in Equation~\eqref{eq:partition}, and the last can always be true by defining the parameter $\delta>0$ suitably. In block length $n$, we will have for any two sequences $x^n,x'^n\in\cX_{\ell^n}$, where $\ell^n=\ell_1\dots\ell_n$ and $\cX_{\ell^n}=\cX_{\ell_1}\times\dots\times\cX_{\ell_n}$,
\begin{equation}\label{eq:renyi_sum_max}
\sum_{i=1}^n D_\alpha(W_{x_i}\|W_{x'_i})\leq n\frac{E_{1,2}}{2}. 
\end{equation}
So, any two sequences $x^n,x'^n\in\cX_{\ell^n}$ in the same partition have maximum R\'enyi entropy sum given by Equation~\eqref{eq:renyi_sum_max}, and we have seen that any two codewords need to have this sum lower bounded by $nE_{1,2}-C$, with $C$ the constant last term in Equation~\eqref{eq:Renyi_converse_bound}. Therefore, for any big enough $n$ (specifically for any $n>n_0 := -2\alpha\log(1-\lambda_{1,2})/[E_{2,1}(\alpha-1)]$), each region of the partition cannot contain the outputs of two different valid codewords. I.e., necessarily, $N\leq L^n$. It only remains to upper-bound the size of such a partition.

Before presenting the general case, let us start with the particular family of channels which have all the output probabilities bounded away from $0$ by some positive constant $\omega>0$. 
Notice that the partition defined above is very natural for these channels, as the problematic region (when we have probabilities close to zero) simply does not appear. 

\begin{theorem}
\label{thm:Stein_preliminar}
Given a channel $W$ and a constant $\omega>0$ such that for all $x\in\cX$ and $y\in\cY$, $W(x|y)\geq\omega$, any DI code $\cC$ for $W$ in the Stein or Sanov regime can have at most a linear scaling rate $R(n)\leq O(1)$. 
\end{theorem}

\begin{proof}
As we have seen, we just need to upper bound the size of the partition. Given some $y\in\cY$ and $x\in\cX_\ell$, notice that the closest probability of $y$ for an $x'\in\cX_{k\neq\ell}$ is at most $W_x(y)/(1+\delta)$ by virtue of Equation~\eqref{eq:partition}, and the closest probability for an $x''\in\cX_{j\neq k,\ell}$ is $W_x(y)/(1+\delta)^2$. As a matter of fact, we will find valid $x^{(i)}$ values in $L$ different partitions until $W_x(y)/(1+\delta)^L<\omega$, as the channel does not have probabilities that small. Then, we can upper bound the number of allowed partitions $L$ with
\begin{equation}\label{eq:stein_code_size}
\omega\leq\frac{1}{(1+\delta)^L}\, \implies\,
L\leq-\frac{\log\omega}{\log(1+\delta)}=:A(\omega,\delta).
\end{equation}
Where we have used that, for any pair $(x,y)$, the probability $W_x(y)\leq1$ in the expression on the left, and simple algebra to find the expression on the right, where we have also defined the positive constant $A(\omega,\delta)>0$. The code size is then bounded as claimed: $|\cC|= N\leq L^n\leq 2^{n\log A}$, and the rate thus $R(n)=\frac{1}{n}\log N \leq \log A(\omega,\delta)=O(1)$.
\end{proof}

This result shows that for channels with all possible output probabilities uniformly bounded away from $0$, the code length scaling in the Stein or Sanov error regimes (that is, when we only force one error type to go exponentially fast to zero) can only be linear. Next, we study the general case, lifting the restriction on the minimum size of the output probabilities. We find that the superlinear scaling of order $n\log n$ (linearithmic) is excluded both for the Stein and Sanov error regimes.

\begin{theorem}\label{thm:Stein_general}
Any code $\cC$ in the Stein or Sanov regime for DI over an unrestricted channel $W$ can have at most a rate scaling as $R(n)\leq O(\log\log n)$. 
\end{theorem}
\begin{proof}
To be able to use the proof idea of Theorem \ref{thm:Stein_preliminar}, we have to somehow avoid the very small probability values $W_x(y)$. To this end, we 
start by defining a new channel $V$ approximating $W$ such that for all $x\in\cX$ and $y\in\cY$, 
\begin{equation}
\label{eq:neq_channel_V}
V_x(y)=\left\{\frac{1}{K}W_x(y)\quad\text{if}\quad W_x(y)\geq\frac{\delta}{n|\cY|},
\atop0\hspace{1.5cm}\text{if}\quad W_x(y)<\frac{\delta}{n|\cY|},\right.
\end{equation}
with $\displaystyle K = \!\!\!\! \sum_{y:W_x(y)\geq\frac{\delta}{n|\cY|}}\!\! W_x(y) \geq 1-\frac{\delta}{n}$ a normalization constant. Here, $0<\delta<1$ is a parameter to be fixed later.

This transformed channel will have different errors $\lambda_1'$ and $\lambda_2'$ which we need to control in the Stein or Sanov regime. In fact, notice that the particular condition in Equation~\eqref{eq:neq_channel_V} has been chosen such that 
\begin{equation}
\begin{split}
\frac12\|W_{x^n}-V_{x^n}\|_1&\leq\frac12\sum_{i=1}^n\|W_{x_i}-V_{x_i}\|_1\\
&=\frac12\sum_{i=1}^n\sum_{y\in\cY}|W_{x_i}(y)-V_{x_i}(y)|\\
&\leq \frac12\sum_{i=1}^n\sum_{y\in\cY}\frac{\delta}{n|\cY|}=\frac{\delta}{2},
\end{split}
\end{equation}
where the first inequality follows from the subadditivity of the total variation distance under tensor products (triangle inequality), and the last from the bound we have chosen for our modified channel in Equation~\eqref{eq:neq_channel_V}. Now, letting $u_j\neq u_k\in\cX^n$ be two different words from a good DI code for $W$, and $\cE_j\subset\cY^n$ the decision region for the word $u_j$, we get:
\begin{align}
\lambda_1'&=\sup_{j\in[N]}V_{u_j}(\cE_j^c)\leq\sup_{j\in[N]}W_{u_j}(\cE_j^c)+\frac{\delta}{2}\leq\lambda_1+\frac{\delta}{2}, \label{eq:lambda_1_Big}\\
\lambda_2'&=\!\!\!\sup_{j\neq k\in[N]}\!\!\!V_{u_j}(\cE_k)\leq\!\!\!\sup_{j\neq k\in[N]}\!\!\!W_{u_j}(\cE_k)+\frac{\delta}{2}\leq\lambda_2+\frac{\delta}{2}\label{eq:lambda_2_Big}.
\end{align}
When $\lambda_{1,2}<1$ and $\delta$ is small enough, the errors of the transformed channel $V$ maintain the bounded nature $\lambda'_{1,2}<\lambda_{1,2} + \frac{\delta}{2} < 1$. However, the exponential bound $\lambda_{1,2}<2^{-nE_{1,2}}$ is lost, and we have to be more careful in this regime. For this, observe the following letter-wise domination for all $x,y$: $V_x(y)\leq \frac{1}{1-\frac{\delta}{n}}W_x(y)$. Hence, 
\begin{equation}
  V_{x^n}(y^n) 
    \leq \left(\frac{1}{1-\frac{\delta}{n}}\right)^n W_{x^n}(y^n) 
    \leq e^{2\delta} W_{x^n}(y^n),
\end{equation}
where we have used $\delta<1$.
For any subset $\cE\subset\cY^n$ we thus find $V_{x^n}(\cE)\leq e^{2\delta} W_{x^n}(\cE)$, so
\begin{align}
\lambda_1'&=\sup_{j\in[N]}V_{u_j}(\cE_j^c)\leq\sup_{j\in[N]}e^{2\delta} W_{u_j}(\cE_j^c)\leq e^{2\delta}\lambda_1, \label{eq:lambda_1_small}\\
\lambda_2'&=\!\!\!\sup_{j\neq k\in[N]}\!\!\!V_{u_j}(\cE_k)\leq \!\!\!\sup_{j\neq k\in[N]}\!\!\! e^{2\delta} W_{u_j}(\cE_k)\leq e^{2\delta}\lambda_2\label{eq:lambda_2_small}.
\end{align}
Therefore, any DI code for $V$ in the Stein error regime is a code for $W$ which, by Equations~\eqref{eq:lambda_1_small} and \eqref{eq:lambda_2_Big}, is also in the Stein regime. Taking Equations~\eqref{eq:lambda_1_Big} and \eqref{eq:lambda_2_small}, we can argue the same for the Sanov regime. 

It only remains to upper-bound the size of a code for the transformed channel $V$. Notice that, if we restrict the input to letters with output probability distributions having the same support $\cY_s\subset\cY$ (with $s\in [S]$ and $S\leq2^{|\cY|}$), the problem effectively becomes the one in Theorem \ref{thm:Stein_preliminar}, so we can reuse the partition strategy with the only difference that now the parameter $\omega$ that bounds all probabilities is a function of $n$, indeed we have imposed $\omega(n)=\frac{\delta}{n|\cY|}$. Directly form Equation~\eqref{eq:stein_code_size} we extract:
\begin{equation}
 L \leq -\frac{\log\omega(n)}{\log(1+\delta)}
   =\frac{\log n-\log(\delta/|\cY|)}{\log(1+\delta)}=O(\log n).
\end{equation}
We can finally bound the size of the code by $L^n$ times the total number of supports in $\cY^n$, which is bounded by $|S|^n\leq2^{n|\cY|}$:
\begin{equation}
  N \leq L^n2^{n|\cY|}
    =    2^{n(|\cY|+\log L)}
    \leq 2^{O(n\log\log(n))},
\end{equation}
and therefore the rate $R(n) = \frac1n\log N \leq O(\log\log n)$.
\end{proof}

\begin{remark}\label{remark:overcount}
The last steps in the proof above may represent a huge overcount. Indeed, it is highly improbable that different codes constructed for each support can produce a reliable general code for the overall channel without losing any word (just a sequence is received, so there is no way of knowing the support of the probability distribution that produced it). However, notice that this does not affect the size of the code at leading order, as $L$ scales with logarithmically in $n$, while the number $S$ of supports is a constant.

Furthermore, however, we do not know if the bound $R\leq O(\log\log n)$ is tight. The double logarithmic scaling could be an artifact of the proof. Indeed, we have no code constructions for the Stein or Sanov error settings that achieve such a scaling, and it is possible that $R\leq O(1)$ holds generally, as for the special case treated in Theorem \ref{thm:Stein_preliminar}. We leave this open problem for future investigation.
\end{remark}

The results in this section show that having one of the two errors vanishing slowly is not enough to recover linearithmic rates, we actually need both of them to be slow. This shows, once more and in an even more dramatic manner, the fragility of the superlinear behaviour which, while being a characteristic of DI codes, appears only in a very specific, extreme error setting.

\section{Extension to quantum channels}
\label{sec:quantum}
Here we prove how the main results in Section \ref{sec:bounding} can be extended to classical-quantum (cq)-channels and to general quantum channels under the restriction that we only use product states when encoding. For the sake of brevity, we refer the reader to \cite[Sect.~6]{CDBW:DI_classical} for all the necessary quantum information tools (alternatively, see \cite{Wilde-book,holevo:stat-structure-book,CK:book2011}) needed for the following, as well as for a brief review on identification via quantum channels \cite{loeber:PhDThesis,AW:StrongConverse,Winter:QCidentification,Winter:QC-ID-randomness,Winter:Review,CDBW:ID-0-sim,CDBW:ID-0-sim:ICC}. 

Let us start studying the reliability functions for DI over cq-channels $W:\cX\rightarrow \cS(B)$, which have a classical letter (a number $x\in\cX$) as input and a quantum state $W_x\in\cS(B)$ as output. The $n$-extension $W^n:\cX^n\rightarrow\cS(B^n)$ maps input code words $x^n\in\cX^n$ to $W_{x^n}=W_{x_1}\ox\dots\ox W_{x_n}$. 

For the direct part, we just need to define an informationally complete POVM (positive operator-valued measure) a collection of positive semidefinite operators $T=(T_y:y\in\cY)$, such that \(\sum_y T_y = \mathds{1}\). Then, we can form the concatenated channel $\overline{W}=\cT\circ W$ of the cq-channel $W$ followed by the quantum-classical (qc)-channel $\cT(\sigma)=\sum_y(\Tr\sigma T_y)\ket{y}\bra{y}$, which effectively outputs the probability of each outcome $y$ when the POVM $T$ is used as measurement. Now, since that the channel $\overline{W}$ can be identified with the classical channel $\overline{W}(y|x)=\Tr W_xT_y$ and $\widehat{\cX}:=\overline{W}(\cX)\subset\cP(\cY)$ its image of output distributions we can reduce the problem directly to the setup of Theorem \ref{thm:coding} to get the following:

\begin{theorem}
\label{thm:cq-coding}
For any $0<t<1$, and $E(n)>0$, there exists a simultaneous DI code over the cq-channel $W:\cX\rightarrow\cS(B)$ (and its associated $\overline{W}$ as above) on block length $n$ with error exponents $E_1(n)\geq E(n)-\frac1n$ and $E_2(n)\geq E(n)-\frac3n$, and its rate lower-bounded 
\begin{align}
    R(n) &\geq (1-t)\log\left[\Pi_{\sqrt[4]{\frac{6E(n)}{ct^2}}} \left(\!\sqrt{\!\widehat{\cX}}\right)\right]\\
    &\phantom{\geq.} - H(t,1-t) - O\left(\frac{\log n}{n}\right).\nonumber\qed
\end{align}
\end{theorem}

With the same arguments, we can also extend Corollary \ref{cor:coding} and the improved lower bound in Equation~\eqref{eq:improvement_coding} to cq-channels obtaining similar results, albeit changing the lower (for the corollary) or upper (for the improvement) Minkowski dimension of $\sqrt{\!\widetilde{\cX}}$ to that of $\sqrt{\!\widehat{\cX}}$. We remark that these are always rates for \emph{simultaneous} deterministic identification because we measure the output quantum states with the fixed POVM $T^{\ox n}$. See \cite{qhtl-arXiv} for a non-simultaneous code construction that can achieve higher rates in terms of $d_M\!\left(\!\!\sqrt{\!\widetilde{\cX}}\right)$.

We can also extend the results to general quantum channels $\cN:A\rightarrow B$ with an input restriction to a subset $\cX\subset\cS(A)$ which can be effectively reduced to a cq-channel $W:\cX\rightarrow\cS(B)$, $W_x=\cN(x)$. The resulting code is defined as a quantum DI $\cX$-code (see \cite[Sect.~6]{CDBW:DI_classical}) and the corresponding lower bound on the simultaneous rate-reliability function is given, thanks to the reduction from quantum to classical-quantum channels, by Theorem \ref{thm:cq-coding}.

For the converse, we can follow closely the proof of Theorem \ref{thm:converse}. We have to be careful, though, as in the quantum setting we can no longer use the total variation and Euclidean distances to calculate the separation between outputs. We use instead their well-known generalization to quantum states: the \emph{trace} and the \emph{Hilbert-Schmidt distances}, based on the Schatten-$1$ and -$2$ norms, respectively. Given two density matrices $\rho$ and $\sigma$, these distances are related as follows \cite[Lemma 6.2]{CDBW:DI_classical}:
\begin{equation}
  \frac12\|\rho-\sigma\|_1
    \leq \left\|\sqrt{\rho}-\sqrt{\sigma}\right\|_2
    \leq \sqrt{2}\sqrt{\|\rho-\sigma\|_1},
\end{equation}
and the Fuchs-van-de-Graaf inequalities in Equation~\eqref{eq:Classical_FvdG} also hold \cite{FVdG:ineq}. Notice that by denoting the output quantum states as $W_x\in\cS(B)$, the notation of the classical stochastic distances or fidelities between probability distributions is equivalent to the corresponding ones in the quantum setting, and we can actually follow the converse proof step by step. We fix a maximum possible fidelity between output states corresponding to code words from a good DI code, create a covering of $\sqrt{\!\widetilde{\cX}}:=\left\{\sqrt{\rho}:\rho\in\widetilde{\cX}\right\}\subset\cS(B)$ such that any two states in the same ball have a fidelity higher than the initial requirement, and finally count the cardinality of such a covering. We obtain:

\begin{theorem}
\label{thm:cq-converse}
For any DI code over the cq-channel $W:\cX\rightarrow\cS(B)$ on block length $n$ with positive error exponents, $E_1(n),\, E_2(n) \geq E(n) > 0$, the rate is upper-bounded
\begin{equation}
  R(n) \leq \log\left[\Gamma_{\frac12{\sqrt{1-e^{-E(n)/2}}}}\left(\!\sqrt{\!\widetilde{\cX}}\right)\right].\qed\end{equation}
\end{theorem}

Corollary \ref{cor:converse} and the improved bound in Equation~\eqref{eq:improvement_converse} are similarly extended to the quantum setting using again the reduction from general quantum channels with an input state restriction $\cX\subset\cS(A)$ to cq-channels. 

\medskip
Finally, we sketch how the results from Section \ref{sec:steins} extend to cq-channels with the smallest eigenvalue of the $W_x\in\cS(B)$ universally bounded away from $0$, i.e.~$W_x\geq\omega\1$ for all $x\in\cX$. 
The starting point with the hypothesis testing relative entropy and the comparison with R\'enyi relative entropies is basically the same, both in the Stein and the Sanov setting (see \cite{CMW16:Dh-to-D_alpha}): as before, for any two distinct code words $x^n = u_j \neq u_k = {x'}^n$ of a DI code, which may be non-simultaneous, and $\alpha>1$,
\begin{align}
\label{eq:quantum-Stein-Sanov-hypo:alpha}
  nE_{2,1} &\leq D_h^{\lambda_{1,2}}(W_{u_j}\|W_{u_k}) \nonumber\\
        &\leq \widetilde{D}_\alpha(W_{u_j}\|W_{u_k}) - \frac{\alpha}{\alpha-1}\log(1-\lambda_{1,2}) \nonumber\\
       &\leq \sum_{i=1}^n \widetilde{D}_\alpha(W_{x_i}\|W_{x_i'}) - \frac{\alpha}{\alpha-1}\log(1-\lambda_{1,2}), 
\end{align}
where 
\begin{equation}
\widetilde{D}_\alpha(\rho\|\sigma) 
  = \frac{1}{\alpha-1}\log\Tr\left( \sigma^{\frac{1-\alpha}{2\alpha}} \rho \sigma^{\frac{1-\alpha}{2\alpha}} \right)^\alpha
\end{equation}
is the sandwiched R\'enyi relative entropy. 
Next, we use the following continuity bound for the latter, leveraging the eigenvalue lower bound:
\begin{lemma}[{Cf.~Bluhm~\emph{et~al.}~\cite[Cor.~5.7]{Bluhm:Renyi-continuity}}]
\label{lemma:bluhm-bound}
For $\alpha>1$ and  two states $\rho$ and $\sigma$ with $\rho,\,\sigma \geq \omega\1 > 0$ and $\frac12\|\rho-\sigma\|_1=t$, 
\begin{align}
  \widetilde{D}_\alpha(\rho\|\sigma) 
    &\leq \log(1\!+\!t) + \frac{1}{\alpha-1}\log\!\left(\! 1 \!+\! t\omega^{1-\alpha} \!+\! \frac{t^\alpha}{(1+t)^{\alpha-1}} \!\right)\nonumber \\
    &\leq \frac{t}{\ln 2}\left[ 1 + \frac{1}{\alpha-1} \left(\frac{1}{\omega^{\alpha-1}} - \left(\frac{t}{1+t}\right)^{\alpha-1} \right) \right]\nonumber \\
    &\leq t\,K(\omega,\alpha).
\end{align}
\end{lemma}
Using this, we can convert the above Equation~\eqref{eq:quantum-Stein-Sanov-hypo:alpha} into bounds on sums of trace distances: 
\begin{equation}
\label{eq:hypo-tracenorm-bound}
\sum_{i=1}^n \frac12\|W_{x_i}-W_{x_i'}\|_1 
  \geq \begin{cases}
         \frac{E_2}{K(\omega,\alpha)}n + \frac{\alpha\log(1-\lambda_1)}{(\alpha-1)K(\omega,\alpha)}, \\[2ex]
         \frac{E_1}{K(\omega,\alpha)}n + \frac{\alpha\log(1-\lambda_2)}{(\alpha-1)K(\omega,\alpha)}.
       \end{cases}
\end{equation}
Either right hand side has the form $\gamma n - \Delta$. Now, partition $\widetilde{\cX} = \bigcup_{\ell=1}^L \widetilde{\cX}_\ell$ such that the diameter of each $\widetilde{\cX}_\ell$ (with respect to the trace distance) is $\leq\gamma/2$ -- this could for example come from a $\gamma/4$-covering, which has cardinality $L = \Gamma_{\gamma/4}(\widetilde{\cX})$. This means that for any $x^n,\,{x'}^n$ such that $W_{x^n},\,W_{{x'}^n} \in \widetilde{\cX}_{\ell^n} = \widetilde{\cX}_{\ell_1}\times\cdots\times\widetilde{\cX}_{\ell_n}$ (that is, from the same region of the $n$-letter partition), 
\(
  \sum_{i=1}^n \frac12\|W_{x_i}-W_{x_i'}\|_1
    \leq \frac{\gamma}{2}n, 
\)
which is $<\gamma n - \Delta$ for all sufficiently large $n$. In other words, each $\widetilde{\cX}_{\ell^n}$ can contain at most one $W_{u_j}$, thus $N \leq L^n$, or $R(n) \leq \frac1n\log N \leq \log L$. Note that $L$ is a function of $\gamma$, which in turn is determined by $E_{1,2}$, $\alpha>1$ and $\omega>0$, but in any case a constant with respect to $n$.
Thus we proved the following quantum generalisation of Theorem \ref{thm:Stein_preliminar}: 
\begin{theorem}
\label{thm:Stein_cq}
Given a cq-channel $W$ and a constant $\omega>0$ such that for all $x\in\cX$, $W_x\geq\omega\1$, any DI code $\cC$ for $W$ in the Stein or Sanov regime can have at most a linear scaling rate $R(n)\leq O(1)$. 
\qed
\end{theorem}

On the other hand, without this eigenvalue lower bound there are examples of cq-channels attaining linearithmic DI rates in the Sanov setting with $\lambda_1=0$ (i.e.~$E_1=+\infty$) and $\lambda_2\rightarrow 0$ \cite[Sect.~V]{qhtl-arXiv}. Crucially, the codes constructed there are non-simultaneous. We do not know if linearithmic rates are possible for cq-channels in the Stein setting.


\section{Conclusions}
\label{sec:conclusions}
We have initiated the study of the rate-reliability tradeoff between the message length and the exponents of the two error probabilities occurring in DI. The main finding is that, unlike in the communication problem, imposing exponentially (in the block length $n$) small errors of first and second kind makes the slightly superlinear scaling of the message length disappear, and replaces it with a ``regular'' linear-scale rate. Even in the case where only one of the two errors vanishes exponentially, superlinearity is lost. We do not have exact formulas for the rate-reliability function, whose derivation remains as the main open question for further research. But our results are strong enough to show that for sufficiently small error exponents a universal behaviour manifests itself: the rate is essentially proportional to the Minkowski dimension of the set $\sqrt{\!\widetilde{\cX}}$ (up to a factor between $\frac14$ and $\frac12$) multiplied by the logarithm of the inverse exponent. 

Technically, we derive upper and lower bounds on the DI rate for fixed reliability exponents and at finite block length. While the proofs are mainly inspired by those in \cite{CDBW:DI_classical}, the analysis in the non-asymptotic regime allows us to separate the information theoretic from the geometric aspects of the problem, resulting in the bounding of the DI rate in terms of error-exponent-dependent packings (for the lower bound, Theorem~\ref{thm:coding}) and coverings (for the upper bound, Theorem~\ref{thm:converse}). The geometric inspection of the bounds in the regime of small errors yields the appearance of the lower and upper Minkowski dimensions featuring in the main results of \cite{CDBW:DI_classical}. There, these two aspects appeared intertwined, as the covering and packing radii were made to directly depend on the block length (which is treated asymptotically), perhaps obscuring the main ideas behind the proof. 
The finite block length approach started in this paper brings this still very abstract but exciting theory closer to potential implementation. 

The present results go some way towards explaining the origin of the surprising superlinear performance characteristic of deterministic identification: here we see that it has its roots in a sufficiently slow convergence of both errors to zero. By decoupling the error exponents from the block length, we discover that they determine the geometric scale at which the packing and covering properties of $\sqrt{\!\widetilde{\cX}}$ determine the (linear) DI rate. Thus, the Minkowski dimension emerges in the regime of small error exponents. 

We furthermore show how to approach the special case of channels with a vanishing Minkowski dimension of the output probability set through a couple of examples, and extend the reliability results to classical-quantum channels and quantum channels under some input restriction. 
Thirdly, we approach the case of large error exponents for which our lower bound becomes trivial, and comment on some characteristics that the reliability curve must have in that regime, leaving a more thorough study for future work.

While the results up to that point make sense mainly in the regime of positive and comparable error exponents, we go beyond this setting in Section~\ref{sec:steins}, where the regime of extremely differently behaving errors is investigated: one error vanishes exponentially fast while the other is just bounded by a constant. Also in this case (which is closely related to the settings of Stein's Lemma and Sanov's Theorem in Hypothesis testing) we find that the superlinear scaling of DI codes is lost.
These sections complete our first reliability study of deterministic identification making it a parallel to the capacity study \cite{CDBW:DI_classical}, in the sense that every capacity result there can now be obtained through the reliability results in this paper, in particular error regimes. 



While our analysis is based on entropy typicality arguments, which are more naturally extendable to continuous and quantum settings, we also demonstrate that it recovers known results previously derived using the method of types and strong typicality \cite{SPBD:DI_power} (see Example~\ref{ex:DMC}).
It is well-known that in classical channel coding over discrete memoryless channels the method of types can be used to derive optimal error exponents. But such techniques are not applicable in our more general setting, which includes continuous and quantum channels where empirical distributions (types) are not well defined. We think that the non-matching upper and lower bounds on the DI capacity cannot simply be attributed to the use of the ``wrong'' typicality concept. 

As already mentioned, the biggest open problem is closing the gap between our upper and lower bounds on the rate-reliability functions and capacity results. Indeed, our bounds could be taken to suggest that the linearithmic capacity of DI over general channel is $\dot{C}_\text{DI}(W)=\alpha\underline{d}_M\!\left(\!\!\sqrt{\!\widetilde{\cX}}\right)$, with $\alpha\in[\frac{1}{4},\frac12]$ a constant which we still need to fix. 
The recent result that for Gaussian channels $\alpha\geq\frac38$ \cite{galaxy-codes}, together with the difficulty of lowering the upper bound, might be taken to indicate that our code construction (packing with hypothesis testing lemma), and hence the achievability bound, can be improved. Furthermore, we have examples of particular channels (both classical \cite[Sect.~V-F]{CDBW:DI_classical} and quantum \cite{qhtl-arXiv}) for which suitable decoding sets beyond the entropy typical set can be devised which attain the capacity upper bound.


\renewcommand*{\bibfont}{\footnotesize}
\printbibliography

\end{document}